\documentclass[a4paper,USenglish,cleveref, autoref, numberwithinsect]{lipics-v2021}
\usepackage[utf8]{inputenc}
\usepackage{amsmath,amssymb}
\usepackage{amsthm}
\usepackage{stmaryrd}
\usepackage{xcolor}
\usepackage{colonequals}
\usepackage{subcaption}
\usepackage{mathabx}
\usepackage{bm}

\nolinenumbers
\hypersetup{
    colorlinks,
    linkcolor={red!50!black},
    citecolor={blue!50!black},
    urlcolor={blue!30!black}
}

\usepackage[bibliography=common]{apxproof}

\hideLIPIcs
\newtheoremrep{theorem}{Theorem}[section]
\newtheoremrep{corollary}[theorem]{Corollary}
\newtheoremrep{proposition}[theorem]{Proposition}
\newtheoremrep{definition}[theorem]{Definition}
\newtheoremrep{lemma}[theorem]{Lemma}
\newtheoremrep{remark}[theorem]{Remark}
\newtheoremrep{conjecture}[theorem]{Conjecture}

\newcommand{\alphabet}{\Sigma}
\newcommand{\alphabetstar}{\alphabet^{*}}
\newcommand{\e}[1]{\emph{#1}}

\newcommand{\mspan}[2]{\ensuremath{[#1,#2\rangle}}
\newcommand{\vars}{\mathsf{Vars}}

\newcommand{\emptyword}{\epsilon}

\newcommand{\bool}{\textrm{Bool}}
\newcommand{\markers}{\textsf{markers}}
\newcommand{\vop}[1]{\mathop{#1{\vdash}}}
\newcommand{\vcl}[1]{\mathbin{{\dashv}#1}}

\newcommand{\calA}{\mathcal{A}}
\newcommand{\NP}{\textsf{NP}}

\newcommand{\domrela}[1]{\preccurlyeq_{\mathit{#1}}}
\newcommand{\domrule}[1]{D_{\mathit{#1}}}
\newcommand{\domoper}[1]{\eta_{#1}}

\def\conR{\cdot}
\def\mcomp{\sim}

\newcommand{\dom}{\mathsf{dom}}
\newcommand{\SVars}{\mathsf{SVars}}
\newcommand{\Spans}{\mathsf{Spans}}
\newcommand{\Vars}{\mathsf{Variables}} %
\newcommand{\Maps}{\mathsf{Maps}}

\newcommand{\dagg}[1]{#1^\dagger}

\newcommand{\restrict}[1]{|_{#1}}

\newcommand{\sky}{\mathrm{sky}}
\newcommand{\chunk}{\mathrm{chunk}}

\newcommand{\valid}{\mathrm{valid}}
\newcommand{\mask}{\mathrm{mask}}

\renewcommand{\epsilon}{\varepsilon}
\renewcommand{\phi}{\varphi}

\newcommand{\NN}{\mathbb{N}}

\title{Skyline Operators for Document Spanners}
\author{Antoine Amarilli}{LTCI, Télécom Paris, Institut polytechnique de Paris, France \and \url{https://a3nm.net/}}{antoine.amarilli@telecom-paris.fr}{https://orcid.org/0000-0002-7977-4441}{}
\author{Benny Kimelfeld}{Technion - Israel Institute of Technology,
Haifa, Israel}{bennyk@cs.technion.ac.il}{https://orcid.org/0000-0002-7156-1572}{}
 \author{Sébastien Labbé}{École normale supérieure}{}{}{}
 \author{Stefan Mengel}{Univ.\ Artois, CNRS, Centre de Recherche en
 Informatique de Lens (CRIL)}{}{}{}

\authorrunning{A. Amarilli, B. Kimelfeld, S. Labbé, S.Mengel}
\funding{Amarilli, Labbé, and Mengel were partially supported by the ANR project EQUUS ANR-19-CE48-0019.}
\begin{document}

\maketitle

\begin{abstract}
  When extracting a relation of spans (intervals) from a text
  document, a common practice is to filter out tuples of the relation that are deemed
  dominated by others. The domination rule is defined as a partial order that varies along
  different systems and tasks. For example, we may state that a tuple is
  dominated by tuples which extend it by assigning additional attributes, or
  assigning larger intervals. The result of filtering the relation would then be
  the \emph{skyline} according to this partial order. As this filtering may remove most
  of the extracted tuples, we study whether we can improve the
  performance of the extraction by compiling the domination rule into the
  extractor.
  
  To this aim, we introduce the \emph{skyline operator} for
  declarative information extraction tasks expressed as document
  spanners.
  We show that this operator can be expressed via regular operations
  when the domination partial order can itself be expressed as a regular
  spanner, which 
  covers several natural domination rules.
  Yet, we show that the skyline operator incurs a computational cost (under
  combined complexity). First, there are cases where the operator
  requires 
  an exponential blowup on the number of states
  needed to represent the spanner as a sequential variable-set
  automaton. Second, the evaluation may become computationally
  hard. Our analysis more precisely identifies classes of domination rules for which
  the combined complexity is tractable or intractable.
\end{abstract}

\section{Introduction}

The framework of \emph{document spanners}~\cite{faginformal} is an established formalism to express declarative information extraction tasks. A \emph{spanner} specifies the possible ways to assign variables over a textual document, producing so-called \emph{mappings} which are the result of the extraction: each mapping assigns the variables to a factor of the document, called a \emph{span}. The spanner formalism has been defined in terms of several operators, in particular regular operations extended with capture variables (corresponding to so-called \emph{regular spanners}), operators from relational algebra (which can sometimes be translated into regular expressions), string equality (the so-called \emph{core spanners}), etc.

Existing operators do not give a simple way to extract mappings that are \emph{maximal} according to some criteria. However, there are natural cases in which some mappings are preferred to others. In fact, traditional declarative systems for information extraction provide explicit mechanisms for restricting the extracted spans to the maximal ones according to different comparisons. IBM's SystemT~\cite{DBLP:conf/acl/LiRC11} has the \emph{consolidation} rules such as ``contained-within'' (where a span dominates its subspans) and ``left-to-right'' (where a span dominates all shorter spans that begin at the same position). Similarly, the GATE system~\cite{DBLP:conf/anlp/CunninghamHGW97a} 
features \emph{controls} such as ``Appelt'' (which is similar to SystemT's ``left-to-right''). 
Alternatively, in the \emph{schemaless} context of document spanners where we can assign spans to only a subset of variables~\cite{DBLP:conf/pods/MaturanaRV18}, we may want to only capture spans which assign a maximal subset of the variables and cannot be extended by assigning more variables; in the spirit, for instance, of the relational \emph{full disjunction}~\cite{10.1145/191843.191908} or the \texttt{OPTIONAL} operator of SPARQL~\cite{ahmetaj2016challenge}.

To explore the expressive power of operators such as controls and consolidators, Fagin et al.~\cite{DBLP:journals/tods/FaginKRV16} proposed a framework that enriches document spanners with a previous concept of \emph{prioritized repairs}~\cite{DBLP:journals/amai/StaworkoCM12}. There, they defined the notion of a ``denial preference-generating dependency'' (denial pgd) that expresses the binary domination relationship using the underlying spanner language. When this relationship is transitive, the result of applying the denial pgd is precisely the set of maximal tuples. However, they did not address the computational complexity of this operator and, consequently, it has been left open. (Moreover, their study 
does not apply to the schemaless context.)

The notion of maximal matches has been abundantly studied in other
areas of database research, where it is called the \emph{skyline
operator}~\cite{borzsonyi2001skyline}. Intuitively, the skyline of a set of
results under a partial order relation is the set of the results that are
maximal, i.e., are not dominated by another result. The complexity of skyline
computation has been investigated under many dimensions, e.g., I/O
access~\cite{sheng2012worst}, parallel computation~\cite{afrati2012parallel}, or
noisy comparisons~\cite{groz2015skyline}.
However, we are not aware of a study of the complexity of this operator to
extract the maximal matches of document spanners. This is the focus of the
present paper.

\subparagraph*{Contributions.}
We present our contributions together with the structure of the paper. After some
necessary preliminaries (Section~\ref{sec:prelim}),
we first introduce in
Section~\ref{sec:skyline} the skyline operator. The operator is defined as
extracting the maximal mappings of a spanner on a document with respect to
a partial order on the mapping, which we call a \emph{domination relation}.
In particular, we define the \emph{span inclusion}, \emph{span length},
\emph{variable inclusion}, and \emph{left-to-right domination relations}, which
cover the examples presented above.

To allow for a unified study of these operators, and similarly to~\cite{DBLP:journals/tods/FaginKRV16},
we propose a general model where the domination relations are themselves
expressed as document spanners.
More precisely, a \emph{domination rule} is a spanner that defines a domination
relation on every document: it indicates which mappings
dominates which other mappings, by intuitively capturing pairs $(m, m')$ that
indicate that $m'$ dominates~$m$.
We also focus on so-called
\emph{variable-wise rules}, where the domination relation on mappings can be
defined as a product of relations on spans.
In other words, a variable-wise rule is a
spanner expressing which spans dominate which spans, and the domination relation
on mappings is obtained in a pointwise fashion across the variables, 
like the \emph{ceteris paribus}
semantics for preference handling in artificial
intelligence~\cite{BoutilierBDHP04}
or Pareto-optimal points for skyline queries on
multidimensional data~\cite{groz2015skyline}. All examples introduced earlier
can be expressed in this variable-wise way.

We then begin our study of how to evaluate the skyline operator on document
spanners, and start in Section~\ref{sct:closure} with the question of
\emph{expressiveness}: does the operator strictly increase the
expressive power of spanner formalisms, or can it be rewritten using
existing operators?  We show that
\emph{regular
spanners} are closed under the skyline operator, generalizing a result
of~\cite{DBLP:journals/tods/FaginKRV16} to the schemaless context.
By contrast, we show that \emph{core
spanners} are not closed under skylines, even for the fixed variable
inclusion or span inclusion domination relations, again generalizing a result
of~\cite{DBLP:journals/tods/FaginKRV16}.

Next, we explore the question of whether it is possible to tractably
rewrite the skyline operator into regular spanners, to allow for
efficient evaluation like, e.g., the polynomial-time compilation of
the join operator in the schema-based context (see~\cite{phdrgxlog},
Lemma~4.4.7).  We present in Section~\ref{sct:blowups} a lower bound
establishing that this is not the case: even for variable inclusion
domination, applying the skyline operator to a spanner expressed as a
sequential variable-set automaton (VA) incurs a necessary exponential
blowup.
This result is shown by identifying a connection between VAs and \emph{nondeterministic read-once branching programs} (NROBPs). This
general-purpose method can be used outside of the context of skylines, and in
fact we also use it to show a result of independent interest: there are
regex-formulas on which the \emph{natural join} operator incurs an unavoidable
exponential blowup (Theorem~\ref{thm:blowupalgebraic}).

We then move in Section~\ref{sec:complexity} from state complexity to the \emph{computational complexity} of
skyline evaluation for regular spanners. 
This task is clearly tractable in \emph{data complexity}, i.e., for a fixed
spanner and domination rule: we simply compute all captured mappings, and filter
out the non-maximal ones. More interestingly, assuming
$\mathsf{P} \neq \mathsf{NP}$,
we show that the task is
intractable in 
\emph{combined complexity}, i.e., as a function of the input spanner (Theorem~\ref{thm:nphardspanner}),
already in the case of the variable inclusion relation.
Hence, we cannot
tractably evaluate the skyline operator in combined complexity, even without
compiling it to an explicit VA.

Lastly, we study in more detail how the complexity of skyline computation depends
on the fixed domination relation: are there non-trivial domination rules
for which skyline computation is tractable in combined complexity? We show in
Section~\ref{sec:further} a sufficient condition on domination rules
which is satisfied by all example rules that we mentioned and which 
implies (conditional) intractability (Theorem~\ref{thm:maindicho}).
We then show that, for a class of domination rules called
\emph{variable-inclusion-like} rules, a variant of this condition can be used
for a dichotomy to classify which of these rules enjoy
tractable skyline computation (Theorem~\ref{thm:dichovarinclike}). We finish
with examples of tractable and intractable rules in the general case.

We conclude in Section~\ref{sec:conc}. For reasons of space, most
proofs are deferred to the Appendix.

\section{Preliminaries}
\label{sec:prelim}

\begin{toappendix}
	\section{Additional Preliminaries}\label{app:prelim}
	In this section, we give some additional preliminaries for notions that we use in the appendix.

        \subparagraph*{Ref-words.}
  In some proofs, we will work with \emph{ref-words}. Like
  in~\cite{phdrgxlog}, we define \emph{ref-words} to be words over $(\Sigma \cup
  \markers(X))^*$, and define the \emph{ref-word language} of a VA as the set of
  ref-words that it accepts when interpreting it as a usual finite automaton over the alphabet $\Sigma \cup
  \markers(X)$. 
  Note that, for a sequential VA, the ref-words in its
  ref-word language are \emph{valid}, i.e., for each variable, either its
  markers do not appear at all, or exactly one opening marker occurs and is
  followed by exactly one closing marker at some position later in the ref-word.

  \subparagraph*{Trimming.} 
  It will often be convenient in proofs to assume that VAs are \emph{trimmed} in
  the sense that we remove all states that are not part of an accepting run.
  (The case where the automaton has no accepting run is trivial, so we often
  implicitly exclude that case and assume that the resulting VA is well
  defined.) We can trim a VA in linear time, and it does not affect the spanner
  that it defines, nor does it affect sequentiality or functionality. Further,
  it can only make the number of states decrease. Having a trimmed
  sequential VA ensures
  that, for all \emph{partial} runs of the VA, we assign markers in a valid way,
  namely: for each variable, we first have no markers of that variable assigned,
  then we may assign one opening marker for that variable, and then we may
  assign one closing marker for the variable. (By contrast, on sequential VAs
  that are not trimmed, the condition may be violated on partial runs that
  cannot be completed to an accepting run.)

        \subparagraph*{$\bm{\epsilon}$-transitions.}
  In some proofs, we will consider
  a slightly modified
  VA model where we additionally allow \emph{$\varepsilon$-transitions}. Formally,
  the transitions of the VA then include letter transitions, marker transitions, and
  transitions labeled by~$\varepsilon$ that can be taken freely as part of a run.
  Note that this modification clearly does not change the expressive power or
  conciseness of our VA formalisms, because a VA with $\varepsilon$-transitions can
  be rewritten to one without such transitions and with the same number of states in the usual
  way: make final all states having a path of $\varepsilon$-transitions to a
  final state, and add, for every letter or marker transition from a
  state~$q$ to a state~$q'$, for every state $q''$ having a path of
  $\varepsilon$-transitions to~$q$, a transition from~$q''$ to~$q'$ having the same
  label. This translation can be performed in polynomial time, and does not
  change the captured ref-words, so in particular it does not affect the fact
  that the VA is sequential or functional.

\end{toappendix}

\subparagraph*{Languages, spans, mappings, and spanners.}
We fix an
\emph{alphabet}~$\alphabet$ which is a finite set of letters.
A \emph{word} $w$ 
is a finite sequence of letters of~$\Sigma$:
we write $\alphabetstar$ the set of all words.
We write $|w|$ for the length of~$w$ and denote the empty word by $\epsilon$,
with $|\epsilon| = 0$.
A \emph{language}~$L \subseteq \Sigma^*$ is a set of words.
The \emph{concatenation} of two languages $L_1$ and $L_2$
is the language
$L_1 \conR L_2 = \{w_1 w_2 \mid w_1 \in L_1, w_2 \in L_2\}$.
The \emph{Kleene star} of a language $L$ is the language $L^* = \bigcup_{i \in
\NN} L^i$, where we define inductively $L^0 = \{\epsilon\}$ and $L_{i+1} = L
\cdot L_i$ for all $i > 0$.
As usual in the context of document spanners, a \emph{document} is simply a word
of~$\alphabetstar$.
A \emph{span}~$\mspan{i}{j}$ 
is an interval 
$s = [i,j\rangle$ with $0\leq i\leq j$.
Its \emph{length} is $j-i$.
We denote by $\Spans$ the set of all spans.
The \emph{spans} of a document $d$ are the spans $\mspan{i}{j}$ of $\Spans$
with $j \leq |d|$.
We write $d_{[i,j\rangle}$ to mean the contiguous subword of~$d$ at a span
$[i,j\rangle$, 
for example $\text{``qwertyqwerty''}_{[2, 5 \rangle} = 
\text{``qwertyqwerty''}_{[8, 11 \rangle} = 
\text{``ert''}$.
Note that we have $d_{[i,i\rangle} = \epsilon$ for all $0 \leq i \leq |d|$.
A span $\mspan{i}{j}$ is \emph{included} in a span $\mspan{i'}{j'}$ if
$i' \leq i$ and $j' \geq j$.
Two spans \emph{overlap} if there is a non-empty span included in both of
them; otherwise we call them \emph{disjoint}.

We fix an infinite set $\Vars$ of variable names.
A \emph{mapping}~$m$ of a document $d \in \Sigma^*$ is a function 
from a finite set of variables $X \subseteq \Vars$, called the
\emph{domain} $\dom(m)$ of~$m$, to the set of spans of $d$; the variables
of~$\dom(m)$ are said to be \emph{assigned} by~$m$.
We denote the set of all mappings on variables of~$\Vars$ by~$\Maps$.
A mapping~$m$ is called \emph{compatible} with a mapping~$m'$, in symbols
$m \mcomp m'$, if for all $x\in \dom(m) \cap \dom(m')$
we have $m(x) = m'(x)$.

A \emph{spanner} is a function mapping every document $d$ to a finite set of
mappings whose spans are over~$d$, i.e., are included in $\mspan{0}{|d|}$.
For a spanner~$P$, we denote by $\SVars(P)$ the variables appearing in the domain
of at least one of its mappings,
formally $\SVars(P) \colonequals \{x \in \Vars \mid \exists d \in \alphabetstar,
\exists m \in P(d), x \in \dom(m)\}$.
A spanner~$P$ is \emph{schema-based} if 
all its output mappings assign exactly the
variables of $\SVars(P)$, i.e.,
for every $d \in \alphabetstar$ and $m \in
P(d)$, we have $\dom(m) = \SVars(P)$. Otherwise, $P$ is called
\emph{schemaless}~\cite{PeterfreundFKK19},
or 
\emph{incomplete}~\cite{DBLP:conf/pods/MaturanaRV18}.
We say a spanner $P$ 
\emph{accepts} or \emph{captures} a mapping $m \in \Maps$ on a document $d \in \alphabetstar$ 
if $m \in P(d)$.

\subparagraph*{Variable-set automata.}
We focus mostly on the \emph{regular spanners}, that can be expressed using 
\emph{variable-set automata} (or VAs). These
are intuitively nondeterministic automata where each transition is labeled
either by a letter or by a \emph{marker} indicating which variable is opened
or closed.
Formally, for a set $X$ of variables, we denote by $\markers(X)$ the set of \emph{markers}
over~$X$: 
it contains for each variable $x\in X$ the \emph{opening marker} $\vop{x}$
and the \emph{closing marker}~$\vcl{x}$.
Then, a VA on alphabet $\Sigma$ is 
an automaton $\calA = (Q, q_0, F,
\delta)$ where $Q$ is a finite set of \emph{states}, $q_0\in Q$ is the
\emph{initial state}, $F \subseteq Q$ are the \emph{final states}, and
$\delta \subseteq Q \times (\Sigma \cup \markers(X)) \times Q$ is the
\emph{transition relation}: we write the transitions $q \rightarrow^\sigma q'$
to mean that $(q, \sigma, q') \in \delta$. Note that the transitions contain
both \emph{letter transitions}, labeled by letters of~$\Sigma$, and \emph{marker
transitions}, labeled by markers of~$\markers(X)$.

A \emph{run} of~$\calA$ on a document $d \in \Sigma^*$ is a sequence $\rho : q_0
\rightarrow^{\sigma_1} q_1 \cdots q_{n-1} \rightarrow^{\sigma_n} q_n$ such that
the restriction of $\sigma_1 \ldots \sigma_n$ to the letters
of~$\Sigma$ is exactly~$d$; it is \emph{accepting} if we have $q_n \in F$. We say that
$\rho$ is \emph{valid} if, for each variable $x \in X$, either the markers
$\vop{x}$ and $\vcl{x}$ do not occur in $\sigma_1 \cdots \sigma_n$, or they
occur exactly once and $\vop{x}$ occurs before $\vcl{x}$. We say that $\calA$ is
\emph{sequential} if all its accepting runs are valid.
In this paper, we always assume that VAs are sequential, and only speak of VAs
to mean sequential VAs.
The run $\rho$ then defines a mapping~$m$ on~$d$ by intuitively assigning the
variables for which markers are read to the span delimited by these markers.
Formally, we associate to each index $0 \leq k \leq n$ of the run a position
$\pi(k)$ in~$d$ by initializing $\pi(0) \colonequals 0$ and setting $\pi(k+1)
\colonequals \pi(k)$ if the transition $q_k
\rightarrow^{\sigma_{k+1}} q_{k+1}$ reads a marker, and $\pi(k+1) \colonequals
\pi(k) + 1$ if it reads a letter; note that $\pi(n) = |d|$.
Then, for each variable $x$ whose markers are read in~$\rho$,
letting $\sigma_{i} = \vop{x}$ and $\sigma_{j} = \vcl{x}$ with $i < j$
because the run is valid, we set $m(x) \colonequals \mspan{\pi(i)}{\pi(j)}$.

A sequential VA $\calA$ thus defines a spanner $P_\calA$ that maps each document~$d$ to
the set $P_\calA(d)$ of mappings obtained from its accepting runs as we
explained. Note that different accepting runs may yield the same mapping.
We 
sometimes abuse notation and identify VAs with the spanners
that they define.
The \emph{regular spanners} are 
those that can be defined by
VAs, or, equivalently~\cite{DBLP:conf/pods/MaturanaRV18}, by sequential
VAs.
A sequential VA is \emph{functional} if it defines a schema-based
spanner, i.e., every mapping
assigns every variable that occurs
in the transitions of the VA.

\subparagraph*{Regex formulas.}
Our examples of spanners in this paper will be given not as VAs but in the more
human-readable formalism of \emph{regex formulas}.
The \emph{regex formulas} over an alphabet~$\alphabet$ are the expressions defined
inductively from the empty set~$\emptyset$, empty word~$\epsilon$, and single
letters $a \in \alphabet$, using the three regular operators of disjunction
($e_1 \lor e_2$), concatenation ($e_1 e_2$), and Kleene star ($e^*$), along with
\emph{variable captures} of the form $x\{e_1\}$ where~$x$ is a variable. A
regex-formula $r$ on a document~$d\in \alphabetstar$ defines a spanner on the
variables occurring in~$r$. Intuitively, 
every match
of~$r$ on~$d$ yields a mapping where the variables are assigned to well-nested
spans following the captures; see~\cite{faginformal} for details. We require
of regex-formulas  that, on every document $d \in \alphabetstar$, they assign
each variable at most once; but we allow them to define schemaless spanners,
i.e., they may only assign a subset of the variables.

It is known that regex formulas capture a strict subset of the regular spanners;
see~\cite{faginformal} in the case of schema-based spanners and
\cite{DBLP:conf/pods/MaturanaRV18} for the case of schemaless spanners.

\subparagraph*{Cartesian Products.}
Given two spanners $P_1$ and $P_2$ where $X_1 = \SVars(P_1)$ and $X_2 =
\SVars(P_2)$ are disjoint,
the \emph{Cartesian product} $P_1\times P_2$ of $P_1$ and $P_2$ is the spanner
on variables $X_1\cup X_2$ which on every document $d$ captures the mappings
$(P_1\times P_2)(d) \colonequals P_1(d) \times P_2(d)$. Here, we interpret a pair $(m_1, m_2) \in P_1(d) \times P_2(d)$ 
as the merge of the two mappings, i.e., the mapping defined according to~$m_1$
on~$X_1$ and according to~$m_2$ on~$X_2$.
If $P_1$ and $P_2$ are given as sequential VAs, then one can compute in polynomial time a sequential VA for $P_1\times P_2$.

\begin{toappendix}
\subparagraph*{Spanner algebra.}
We here introduce some operators on spanners and their properties.

For every spanner $P$ and every subset $Y\subseteq \SVars(P)$, we define the projection operator $\pi_Y$ by saying that $\pi_{Y} P$ is the spanner that extracts on every document $d$ the set $(\pi_{Y}P)(d) \colonequals \{m|_{Y}\mid m\in P(d)\}$ where $m|_{Y}$ is the restriction of $m$ to $Y$.

The \emph{natural join} $P_1 \Join P_2$ of two spanners $P_1$ and~$P_2$ is a
spanner which 
accepts all the mappings $m$
which are the union of two compatible mappings 
$m_1$ accepted by $P_1$ and $m_2$ accepted by $P_2$. Said differently,
$(P_1 \Join P_2)(d) \colonequals
\{m \in \Maps \mid
\exists m_1 \in P_1(d), \exists m_2 \in P_2(d), m_1 \mcomp m \land m_2 \mcomp m
\land \dom(m_1) \cup \dom(m_2) = \dom(m)
\}$.

We remark that if $\SVars(P_1) \cap \SVars(P_2) = \emptyset$ then the join operator 
is the Cartesian product defined before.

The intersection operator $\cap$ is defined to compute the spanner $P_1\cap P_2$ which
on every document computes the set $(P_1\cap P_2)(d)\colonequals P_1(d) \cap P_2(d)$.
Observe that that if $\SVars(P_1) = \SVars(P_2)$ and both spanners are schema-based, then the join operator 
is the intersection: $(P_1 \Join P_2)(d) = P_1(d) \cap P_2(d)$.

The union $P_1\cup P_2$ is defined to as the spanner which
on every document computes the set $(P_1\cup P_2)(d)\colonequals P_1(d) \cup P_2(d)$.

The \emph{difference}, $P_1 - P_2$ is a binary operator which 
accepts all mappings accepted by $P_1$ which are not accepted by
$P_2$. Said differently,
$(P_1 - P_2)(d) \colonequals P_1(d) \setminus P_2(d)$.
Note that this is the usual difference operator on sets, and \emph{not} the difference
operator defined in~\cite{phdrgxlog} which accepts mappings of~$P_1$ for which
no compatible mapping is accepted by~$P_2$.

It is known that the projection, natural join operator and union operators do not increase the expressive
power of regular spanners, see~\cite{faginformal} for the case of schema-based
spanners and \cite{DBLP:conf/pods/MaturanaRV18} for schemaless
spanners. It follows that the same is true for the Cartesian product
operator. We will show later in Section~\ref{sct:decompose} that intersection does not increase the expressivity, either. As for the difference operator, the same result is proven
in~\cite{faginformal} for schema-based regular spanners, but we are not aware of the same
result for schemaless spanners and for our semantics of difference.

  \subsection{Schemaless Regular Spanners are Closed under Difference}
We show the following:

\begin{propositionrep}
  \label{prp:diffclose}
  The (schemaless) regular spanners are closed under difference: given two regular spanners $P_1$
  and $P_2$, the difference $P_1 - P_2$ can be expressed as a regular spanner.
\end{propositionrep}
Note that, if the spanners $P_1, P_2$ from the proposition are given as sequential VAs,
then the VA constructed for $P_1-P_2$ in the proof of Proposition~\ref{prp:diffclose} is generally exponentially bigger.

\begin{proof}
  Let $\calA_1$ and $\calA_2$ be the two sequential VAs for $P_1$ and $P_2$ respectively, let $X$
  be the set of their variables, and recall that $\markers(X)$ is the set of
  markers of the form $\vop{x}$ or $\vcl{x}$ for $x \in X$. 
  We use the notion of \emph{ref-words} (see Appendix~\ref{app:prelim}). As the VAs are sequential, their accepting runs only capture valid ref-words.  What is more, on a document $d \in \alphabetstar$,
  the set of mappings produced by a
  sequential VA can
  be obtained from the words of its ref-word language whose \emph{erasure}
  (i.e., removing the marker symbols) yields~$d$: for each such ref-word $w$, we
  obtain a mapping that assigns the variables whose markers appear in~$w$, at
  the span defined by the unique opening and closing markers. However, the same
  mapping can be obtained by multiple ref-words, because the order between
  markers may differ. For instance, the ref-words $\vop{x} \vop{y} a \vcl{x}
  \vcl{y}$ and $\vop{x} \vop{y} a \vcl{y} \vcl{x}$ define the same mapping.

  To avoid this, we will normalize the automata (we note that a similar
  normalization is done in the schema-based case in~\cite{faginformal} under the
  name of \emph{lexicographic} VAs). Let $<$ be a total order on
  $\markers(X)$: we impose that all opening symbols come before all closing
  symbols, i.e., $\vop{x} < \vcl{y}$ for each $x, y \in X$, to ensure that
  ref-words remain valid. A sequential VA is \emph{ordered} relative to~$<$ if, for any
  ref-word in its ref-word language, for every contiguous subsequence of
  markers, then they are ordered relative to~$<$.

  It is not hard to see that we can rewrite sequential VAs, up to an exponential blowup, to
  ensure that they are ordered. Given an input sequential VA $\calA = (Q, q_0, F, \delta)$,
  build the ordered VA $\calA' = (Q \cup Q', q_0, F \cup F', \delta')$ where $Q'$ is a
  primed copy of~$Q$ (i.e., $Q' = \{q' \mid q \in Q\}$), where $F'$ is defined
  in the same way, and where $\delta'$ is
  initialized to perform the letter transitions of~$\delta$ but from primed and
  unprimed states to unprimed states: formally, for each letter transition
  $(q_1, a, q_2) \in \delta$, we add two letter transitions $(q_1, a, q_2)$ and
  $(q_1', a, q_2)$ to~$\delta$.

  Now, consider every path in~$\calA$ that traverses only marker transitions, and,
  letting $q_1$ and $q_2$ be the initial and final states, add a path from $q_1$
  to~$q_2$ in~$\calA'$ that goes via fresh states and reads the same markers but in
  the order given by~$<$.

  After this modification, we claim that $\calA'$ accepts the same ref-word language
  as $\calA$. Indeed, any path from an initial to a final state in~$\calA$ can be
  replayed in~$\calA'$: whenever we take marker transitions in~$\calA$, then we follow a
  path in~$\calA'$ bringing us to a primed state which is then indistinguishable
  from the corresponding unprimed state. Conversely, any path in~$\calA'$ can be
  replayed in~$\calA$. Further, $\calA'$ is ordered relative to~$<$: this is because all
  paths of contiguous marker symbols are ordered by construction. (Note that the
  reason why we distinguish between primed and unprimed states is to ensure
  that, once we have taken a path, then we cannot continue with another path
  without taking a letter transition.)

  Coming back to our input VAs $P_1$ and $P_2$, we apply this
  transformation to them, to obtain VAs  $P_1'$ and $P_2'$ that are ordered
  relative to the same order. Now, we can conclude with some standard
  automata-theoretic manipulations on~$P_2$, and then by taking a product
  automaton. Specifically, first modify $P_2$ to be complete, i.e., add a sink state
  $q_\bot$ 
  and ensure that all states have a transition for all letters and markers from
  every state (including $q_\bot$) to~$q_\bot$: note that the resulting VA is
  still sequential because the new state is never part of an accepting run.
  Second, make $P_2$ deterministic as an automaton on $\Sigma \cup \markers(X)$,
  by applying the standard determinization procedure for finite automata. The resulting VA accepts
  the same language of ref-words, hence it is still sequential and ordered
  relative to the same order: but the VA is
  now deterministic in the sense that every ref-word of its language of
  ref-words (hence, every mapping) has precisely one accepting run, and each
  ref-word is the label of a path in the automaton from the initial state (the
  path is generally non-accepting, i.e., it may end in the sink).

  Having modified $P_2'$ in this way, let us construct the product automaton.
  Write $P_1' = (Q_1, q_{0,1}, F_1, \delta_1)$ and
  $P_2' = (Q_2, q_{0,2}, F_2, \delta_2)$. We 
  construct $(Q_1 \times Q_2, (q_{0,1}, q_{0,2}), F_1
  \times (Q_2 \setminus F_2), \delta)$, where we define $\delta$ to do
  transitions in both components, i.e., for each marker or letter~$\ell$, for 
  each $(q_1, \ell, q_1') \in
  \delta_1$ and $(q_2, \ell, q_2') \in \delta_2$, add $((q_1, q_2), \ell, (q_1',
  q_2'))$ to~$\delta$. 

  We claim that the resulting VA is sequential, because the projection of any
  accepting run to the first component yields an accepting run of~$P_1$, which
  is sequential. Further, it is ordered for the same reason.
  Now, we claim that the automaton accepts the ref-word language
  which is the difference of that of~$P_1'$ and that of~$P_2'$: this is because we
  are doing the standard product construction on automata and the
  automaton~$P_2'$ is deterministic. Formally, any ref-word produced by the
  product gives an accepting run for~$P_1'$ of that word and the unique run
  for~$P_2'$ on that ref-word which is non-accepting, hence witnesses that $P_2$
  reject that ref-word because it is deterministic. Conversely, any ref-word
  accepted by~$P_1'$ and rejected by~$P_2'$ gives an accepting run for that
  ref-word in the product, using the fact that $P_2'$ is complete.

  We now conclude because the resulting VA accepts precisely the ref-words
  accepted by~$P_1'$ and rejected by~$P_2'$: as $P_1'$ and $P_2'$ are ordered,
  the mappings accepted by their difference correspond to the difference of the
  ref-word languages, concluding the proof.
\end{proof}
\end{toappendix}

\begin{toappendix}
\subsection{Decomposing VAs into Functional VAs, and Closure Under Intersection}\label{sct:decompose}

Remember that a functional VA $\calA$ is a sequential VA such that every
  mapping $m$ captured by~$\calA$ on any document $d$ assigns a span
  to every variable of the domain of~$\calA$.
  Given a spanner $P$ and a variable set $X$, we let $P^{[X]}$ denote the spanner that on every document $d$ extracts $P^{[X]}(d)\colonequals \{m\in P(d)\mid \dom(m) = X\}$.

The following result will be useful in several places:
\begin{lemma}\label{lem:capturesubsetfunctional}
	For every regular spanner $P$ and every variable set $X$, there is a functional VA defining $P^{[X]}$.
\end{lemma}
\begin{proof}
	Assume that $X\subseteq \SVars(P)$; otherwise the statement is trivial since $P^{[X]}$ captures no mappings on any document. Let $\calA$ be a sequential VA defining $P$. We show how to construct a functional VA for $P^{[X]}$.
	
	In a first step, we construct a VA~$\calA^{\le X}$ by deleting from $\calA$ all marker transitions for variables not in $X$. On any document $d$, the VA $\calA^{\le X}$ captures exactly the set $\{m\in P(d)\mid \dom(m) \subseteq X\}$. Note that $\calA^{\le X}$ is sequential because $\calA$ is. 
	
        In the same way, for every $x\in X$, we construct a sequential VA $\calA^{\le (X\setminus x)}$ capturing $\{m\in P(d)\mid \dom(m) \subseteq X\setminus \{x\}\}$. Then, we can construct a sequential VA $\calA^{< X}$ capturing $\{m\in P(d)\mid \dom(m) \subsetneq X\}$ on every document as follows: by renaming, assume that the state sets of all $\calA^{\le (X\setminus x)}$ are disjoint. Add a new initial state $s$, connect it to the initial states of all $\calA^{\le (X\setminus x)}$ by an $\epsilon$-edge (recall the definition of $\epsilon$-transitions in Appendix~\ref{app:prelim}). Moreover add a new final state $t$ which is connected to all final states of all $\calA^{\le (X\setminus x)}$ by an $\epsilon$-edge. Let $s$ be the initial state of $\calA^{< X}$ and let $t$ be its single final state. It is easy to see that the resulting VA $\calA^{< X}$ is sequential and accepts the desired mappings.
		
	Remember that for every VA $\calA'$, we denote by $P_{\calA'}$ the spanner defined by $\calA'$. Then
	\begin{align*}
		P^{[X]} =  P_{\calA^{\le X}} - P_{\calA^{<X}}.
	\end{align*}
	With Proposition~\ref{prp:diffclose} we get a sequential VA computing $P^{[X]}$. Moreover, since all captured mappings assign spans to all variables in $X$, that VA is functional.
\end{proof}

We get the following direct consequence whose proof is immediate:

\begin{lemma}\label{lem:decomposeregular}
	Let $P$ be a regular spanner. Then for every $X\subseteq \SVars$ there is a functional VA defining $P^{[X]}$ and we thus have $P = \bigcup_{X\subseteq \SVars(P)} P^{[X]}$.
\end{lemma}

  We deduce that regular spanners are closed under intersection. (Note that, unlike the schema-based case, in the schemaless case the intersection operator on sets cannot be immediately expressed using the join operator.)
\begin{lemma}
	\label{lem:intersectionclose}
	The (schemaless) regular spanners are closed under intersection: given two regular spanners $P_1$
	and $P_2$, the difference $P_1 \cap P_2$ can be expressed as a regular spanner.
\end{lemma}
\begin{proof}
We use Lemma~\ref{lem:decomposeregular} to decompose $P_1 = \bigcup_{X\subseteq \SVars(P)} P_1^{[X]}$ and $P_2 = \bigcup_{X\subseteq \SVars(P)} P_2^{[X]}$. Whenever $X_1, X_2\subseteq \SVars(P_1)\cup \SVars(P_2)$ are different, we have that $(P_1^{[X_1]}\cap P_2^{[X_2]})(d) = \emptyset$ on every document, because they $P_i^{[X_i]}$ extract mappings with different domains. It follows that 
\begin{align*}
	P_1 \cap P_2 = \bigcup_{X\subseteq \SVars(P_1)\cap \SVars(P_2)} (P_1^{[X]}\cap P_2^{[X]}).
\end{align*}
Since for every $d$, all mappings in $P_1^{[X]}(d)$ and $P_2^{[X]}(d)$ assign precisely the variables in~$X$, in this specific context the intersection operator can be expressed as a natural join, like in the schema-based setting:
  \begin{align*}
	P_1^{[X]}\cap P_2^{[X]} = P_1^{[X]}\Join P_2^{[X]}.
\end{align*}
Since, as discussed before, regular spanners are closed under joins, we have that for every $\subseteq \SVars(P_1)\cap \SVars(P_2)$ the spanner $P_1^{[X]}\cap P_2^{[X]}$ is regular. The lemma then follows by closure of regular spanners under union.
\end{proof}
\end{toappendix}

\section{The Skyline Operator}\label{sec:skyline}

In this paper, we define and study a new
operator called the \emph{skyline operator}.
Its goal is to only extract mappings 
that contain the maximum amount
of information in a certain sense. 

\subparagraph*{Domination relations.}
We begin by defining \emph{domination relations} which describe how to compare
the information given by two mappings on a given document~$d$.

\begin{definition}
  A \emph{pre-domination relation} $\domrela{}$ for a document~$d$ is a relation on
  the set of mappings $\Maps$ of~$d$.
  We say that it is a \emph{domination relation} if it is a
  (non-strict) partial order, i.e., it is reflexive, transitive, and
  antisymmetric.
  For $m_1, m_2 \in \Maps$, we say that $m_2$ \emph{dominates} $m_1$ if $m_1
  \domrela{} m_2$, and write
  $m_1 \not{\domrela{}} m_2$ otherwise.
\end{definition}

The goal of the domination relation is to define which mappings are preferred to others, intuitively because they 
contain more information; 
it may depend on the document, though we will present many
examples where it does not.

We introduce several domination relations that, as discussed in the Introduction, are part of practical systems and which we consider throughout this
paper:

\begin{example}
  \label{exa:selfdom}
The simplest relation is the trivial \emph{self domination} relation
  $\domrela{self}$ where every mapping only dominates itself, i.e., the pairs in the relation are $(m,m)$ for $m\in \Maps$.
\end{example}

\begin{example}\label{ex:varinc}
  The \emph{variable inclusion relation} $\domrela{varInc}$ contains the
  pairs $(m_1, m_2)$ such that for all $x\in \Vars$, if $m_1(x)$ is defined, then $m_2(x)$ is defined as well and $m_1(x)= m_2(x)$. Intuitively, we have $m_1\domrela{varInc} m_2$ whenever $\dom(m_1) \subseteq \dom(m_2)$ and $m_1 \mcomp m_2$,
  i.e.,
  when $m_2$ is an extension of $m_1$ that potentially assigns more variables
  than~$m_1$.
\end{example}

\begin{example}
  \label{exa:spaninc}
  The \emph{span inclusion relation} $\domrela{spanInc}$ contains the pairs $(m_1, m_2)$ of mappings with the same domain ($\dom(m_1) = \dom(m_2)$)
  such that for every $x\in \dom(m_1)$ the span $m_1(x)$ is included in $m_2(x)$.
  Intuitively, $m_1$ and $m_2$ match the same variables in the same parts of a document, but the matches of variables in $m_1$ are subwords of their matches in~$m_2$.
\end{example}

\begin{example}\label{exa:leftotoright}
	The \emph{left-to-right relation} $\domrela{ltr}$ contains the pairs $(m_1, m_2)$ of mappings with the same domain 	
	such that, for every variable $x$ on which $m_1$ and $m_2$ are defined, the
        spans $m_1(x)$ and $m_2(x)$ start at the same position but $m_2(x)$ is no shorter than $m_1(x)$.
\end{example}

\begin{example}
  \label{exa:spanlen}
The \emph{span length relation} $\domrela{spanLen}$ contains the pairs $(m_1, m_2)$ of mappings with the same domain 
  where for every $x\in \dom(m_1)$ the span $m_2(x)$ is no shorter
  than $m_1(x)$.
  Intuitively, $\domrela{spanLen}$ prefers longer spans over shorter ones,
  anywhere in the document.
\end{example}

\subparagraph*{Domination rules.}
We now introduce \emph{domination rules} which associate to each document~$d$
a domination relation over~$d$. In this paper, we express domination
rules as spanners
on specific domains.  To this end, given
a set of variables $X$, we write $\dagg{X}$ to mean a set of annotated
copies of the variables of~$X$, formally $\dagg{X}\colonequals \{\dagg{x} \mid x \in X\}$.
We extend the notation to mappings by defining $\dagg{m}$ for a mapping~$m$
to be the mapping with domain $\dom(\dagg{m}) = \dagg{\dom(m)}$ such that
for all $x \in \dom(m)$ we have $\dagg{m}(\dagg{x})\colonequals m(x)$.
We then define:
\begin{definition}
    A \emph{pre-domination rule} $\domrule{}$ on a set of variables $X \subseteq \Vars$
    is 
    a (schemaless) spanner with 
    $\SVars(\domrule{}) \subseteq X \cup \dagg{X}$.
    For every document $d \in \alphabetstar$, 
    we see $\domrule{}(d)$ as a pre-domination relation~$\domrela{}$ on~$d$
    defined by the mappings captured
    by~$\domrule{}$ on~$d$, the left-hand-side and right-hand-side of the
    comparability pairs being
    the restrictions of the mappings to~$X$ and to~$\dagg{X}$ respectively.
    Formally, the relation~$\domrela{}$ is:
    $R \colonequals \{(m\restrict{X},m') \mid m \in \domrule{}(d),
    \dagg{(m')} = m\restrict{\dagg{X}}\}
    \}$.

    We say that $\domrule{}$ is a \emph{domination rule} if, on every document
    $d \in \alphabetstar$, the pre-domination
    relation~$R$ defined above is a domination relation, i.e., it correctly
    defines a partial order.
\end{definition}

Intuitively, for every document $d$, the domination rule $\domrule{}$ defines
the domination relation~$\domrela{}$ where each
mapping $m \in \domrule{}(d)$ denotes a pair, i.e., the restriction of $m$ to $X$
is dominated by the restriction of $m$ to $\dagg{X}$ (renaming the variables
from~$\dagg{X}$ to~$X$).
Note that pre-domination rules are just 
an intermediary notion; in the
sequel, we only consider 
domination rules.

\begin{example}
  For any set~$X$ of variables,
  each of the domination relations introduced in 
  Examples~\mbox{\ref{exa:selfdom}--\ref{exa:leftotoright}}
can be defined by a domination rule
  expressed by a regular spanner on~$X$ (for the span length domination relation
  of Example~\ref{exa:spanlen}, see 
  Lemma~\ref{lem:pumping1}). At the end of the section, we explain
  how to express them in a more concise \emph{variable-wise} way that does not depend
  on~$X$.
\end{example}

\subparagraph*{The skyline operator.}
We have introduced domination rules as a way to define domination relations over
arbitrary documents. We can now introduce
the \emph{skyline operator} 
to extract maximal mappings, i.e., mappings that
are not dominated in the domination relation:

\begin{definition} 
  \label{def:skyop}
  Given a domination rule $\domrule{}$, the \emph{skyline operator}
    $\domoper{\domrule{}}$ of $\domrule{}$ applies to a spanner~$P$ and defines
    a spanner $\domoper{\domrule{}}P$ in the following way:
    given a document $d$, 
    writing $\domrela{}$ to denote the domination relation $\domrule{}(d)$ given
    by $\domrule{}$ on~$d$,
    the result of $\domoper{\domrule{}}P$ on~$d$ is
    the set of maximal mappings of~$P(d)$ under the domination
    relation~$\domrela{}$. Formally, we have:
    $(\domoper{\domrule{}} P)(d) \colonequals
    \{m \in P(d) \mid \forall m' \in P(d)
    \setminus \{m\}\colon  m \not{\domrela{}} m'\}$.
\end{definition}

Intuitively, the operator
filters out the mappings that are dominated by another mapping according to the domination relation
defined by the domination rule over the input document.

\begin{example}\label{ex:skyline}
In Figure~\ref{fig:exskyline} we show the effect of the skyline operator with
  respect to some of our example domination relations. Assume that we are given a
  spanner $P$ in variables $\{x,y\}$ that on a given document $d$ extracts the
  mappings given in Figure~\ref{fig:exskylinea} (here a dash~``$-$'' means that the variable is not assigned by a mapping).
  We show the result of applying the skyline operators with (possibly non-regular) domination rules defining
  the variable inclusion domination relation $\domrela{varInc}$ (Figure~\ref{fig:exskylineb}),
  the span inclusion domination relation $\domrela{spanInc}$ (Figure~\ref{fig:exskylinec}),
  and the span length domination relation~$\domrela{spanLen}$ (Figure~\ref{fig:exskylined}).
Note that, for the variable inclusion domination rule, the skyline only makes sense for
schemaless spanners, as two distinct mappings that assign the same variables are always 
incomparable.%
\begin{figure}
\begin{subfigure}[t]{.23\textwidth}
\centering
\begin{tabular}{ll}
$m(x)$&$m(y)$\\
\hline
\mspan{1}{2}& \mspan{2}{3}\\
$-$ & \mspan{2}{3}\\
\mspan{0}{2}& \mspan{2}{3}\\
\mspan{4}{6}& \mspan{4}{10}\\
\end{tabular}
\caption{The extracted mappings $P(d)$.}
\label{fig:exskylinea}
\end{subfigure}
\hfill
\begin{subfigure}[t]{.25\textwidth}
\centering
\begin{tabular}{ll}
	$m(x)$&$m(y)$\\
	\hline
	\mspan{1}{2}& \mspan{2}{3}\\
	\mspan{0}{2}& \mspan{2}{3}\\
	\mspan{4}{6}& \mspan{4}{10}\\
                \null\\
\end{tabular}
  \caption{Skyline under the variable inclusion relation.}
\label{fig:exskylineb}
\end{subfigure}
\hfill
\begin{subfigure}[t]{.23\textwidth}
	\centering
	\begin{tabular}{ll}
		$m(x)$&$m(y)$\\
		\hline
$-$ & \mspan{2}{3}\\
\mspan{0}{2}& \mspan{2}{3}\\
\mspan{4}{6}& \mspan{4}{10}\\
                \null\\
	\end{tabular}
        \caption{Skyline under the span inclusion relation.}
\label{fig:exskylinec}
\end{subfigure}
\hfill
\begin{subfigure}[t]{.23\textwidth}
	\centering
	\begin{tabular}{ll}
		$m(x)$&$m(y)$\\
		\hline
$-$ & \mspan{2}{3}\\
		\mspan{4}{6}& \mspan{4}{10}\\
                \null\\
                \null\\
	\end{tabular}
	\caption{Skyline under the span length relation.}
\label{fig:exskylined}
\end{subfigure}
\caption{Extracted mappings before and applying different skyline operators; see Example~\ref{ex:skyline}.}
	\label{fig:exskyline}
\end{figure}
\end{example}

\subparagraph*{Variable-wise rules.}
We have defined our skyline operator relative to domination rules expressed as
spanners on explicit sets of variables. However,
it will often be convenient to define the rules as
as products of rules on a single variable by applying the product operator.
This 
ensures that the rule is ``symmetric'' in the sense that all variables
behave the same:

\begin{definition}
	Let $\domrule{}$ be a domination rule in a single variable $x$, i.e., a spanner
        using variables of $\{x, \dagg{x}\}$.
        For $y \in \Vars$, we let $\domrule{}^y$ be the domination rule where we
        replace $x$ and $\dagg{x}$ by $y$ and $\dagg{y}$, i.e., on every
        document $d$, the set of mappings $\domrule{}^y(d)$ consists of one mapping $m^y$
        for each mapping $m \in \domrule{}^y(d)$ with $m^y(y)$ and $m^y(\dagg{y})$ 
        defined like $m(x)$ and $m(\dagg{x})$.
        The \emph{variable-wise domination rule} defined by $\domrule{}$ on a variable
        set $X$ is then simply
        $\bigtimes_{y\in X} \domrule{}^y$.
        A domination rule is said to be \emph{variable-wise} if it can be
        expressed in this way.
\end{definition}

We will often leave the set of variables~$X$ implicit, and may abuse notation
to identify domination rules in a single variable with the variable-wise
domination rule that they can define on an arbitrary variable set.

\begin{example}\label{ex:varwise}
  The self domination rule (Example~\ref{exa:selfdom}) is variable-wise, because
        it can be obtained from the following trivial domination rule on a single
        variable:
\begin{align*}
	\domrule{self} = \alphabetstar \dagg{x}\{x\{\alphabetstar\}\}
        \alphabetstar \lor \alphabetstar.
\end{align*}
        The $\alphabetstar$ term above is used to ensure reflexivity and express the vacuous domination relation between the mapping
        where $x$ is not assigned and the mapping where $\dagg{x}$ is not
        assigned.

  The span inclusion domination rule, left-to-right domination rule, and variable inclusion domination rule (Examples~\ref{ex:varinc}--\ref{exa:leftotoright}) are also variable-wise with the single-variable rules:
\begin{align*}
  \domrule{spanInc} & = \alphabetstar \dagg{x}\{\alphabetstar
  x\{\alphabetstar\}\alphabetstar\} \alphabetstar \lor \alphabetstar. \\
  \domrule{ltr} & = \alphabetstar \dagg{x}\{
x\{\alphabetstar\}\alphabetstar\} \alphabetstar \lor \alphabetstar.\\
\domrule{varInc} & =
\alphabetstar \dagg{x} \{ \alphabetstar \} \alphabetstar\lor  \domrule{self}.
\end{align*}
Here, $\alphabetstar \dagg{x} \{ \alphabetstar \} \alphabetstar$ expresses that
  assigning a variable is better than not assigning it.
\end{example}

As for the variable
length domination rule (Example~\ref{exa:spanlen}), it is also
variable-wise, but a standard pumping argument shows that it cannot be 
defined by a regular spanner:

\begin{lemmarep}\label{lem:pumping1}
    The single-variable span length domination rule $\domrule{spanLen}$ is not expressible as a regular spanner.
\end{lemmarep}
\begin{proof}
We use a standard pumping argument.
  By way of contradiction, assume that there is a sequential VA $P$ that defines
  $\domrule{spanLen}$, say on the variables $x$ and~$\dagg{x}$.
  Let $n$ be the number of states of $P$.
  Then, on input $a^{2n+3}$, the spanner $P$ must extract the mapping $m$ with $m(x)= \mspan{0}{n+1}$ and $m(\dagg{x})= \mspan{n+1}{2n+3}$. Since $|m(x)| > n$, 
  in an accepting run for this mapping, between opening $x$ with $\vdash x$ and closing it again with $\dashv x$, 
  the automaton must visit one state twice.
  Let $i$ be the length of the cycle between these visits. Then on input $a^{2n+3+2i}$, we can follow this cycle $3$ times instead of once, so $P$ extracts $m'$ with $m'(x)= \mspan{0}{2i+ n+1}$ and $m'(\dagg{x})= \mspan{2i+n+1}{2n+3+2i}$. However, since $i \ge 1$, we have that $|m'(x)| = 2i+n+1 > 2n+2 = |m'(\dagg{x})|$, so the mapping~$m'$ should not have been extracted. This contradicts the assumption that $P$ defines $\domrule{spanLen}$, and concludes the proof.
\end{proof}

\section{Closure under the Skyline Operator}\label{sct:closure}

We have defined the skyline operator relative to domination rules expressed by
regular spanners.
One natural question is then to understand
whether the skyline operator under such rules extends the expressive power of
spanner formalisms, or whether it can be defined in existing models. This is
what we investigate in this section.

\subparagraph*{Regular spanners.}
We first focus on regular spanners, and show that they are closed under the
skyline operator for domination rules expressed as regular spanners. We do so by
showing how the skyline operator can be expressed with operations under which
regular spanners are closed, 
namely join, intersection and difference
(see Appendix~\ref{app:prelim} for definitions).

\begin{toappendix}
  \subsection{Proof of Theorem~\ref{thm:closure}}
\end{toappendix}

\begin{theoremrep}\label{thm:closure}
	There is an algorithm that, given a sequential VA defining a regular spanner $P$ and a
        sequential VA defining a domination rule $\domrule{}$, computes
        a sequential VA for $\domoper{\domrule{}}P$.
\end{theoremrep}

\begin{proof}
	For the proof, we will use the operators on spanners introduced in Section~\ref{app:prelim} and the fact that, as discussed there, regular spanners are closed under all of the operators. 
	
    Let $P$ be a spanner defined by a sequential VA and let $\domrule{}$ be the spanner defining the domination relation.
    We define 
    $P_{\sky} = P - \pi_X ((P \times \dagg{P}) \cap (\domrule{} - \domrule{self}))$.
    Intuitively, the spanner $\domrule{} - \domrule{self}$
    extracts all the strict domination pairs (i.e., those which have two distinct mappings),
    and the spanner $P \times \dagg{P}$ extracts all the pairs 
    where both mappings would be extracted by $P$.
    Here $\dagg{P}$ denotes the spanner obtained from~$P$ that captures the
    mappings $\dagg{m}$ for each mapping~$m$ captured by~$P$.
    Then the intersection of these two spanners and projection onto $X$ lets us
    obtain the dominated mappings of $P$. Removing them leaves exactly the
    non-dominated mappings extracted by $P$, so $P_\sky$ extracts exactly the
    mappings in~$\domoper{\domrule{}}P$.

    It remains to show that $P_\sky$ is a regular spanner. But this is immediate from the fact that it is constructed from regular spanners by applying spanner operations under which, as we already discussed above, regular spanners are closed.
  \end{proof}

  \cref{thm:closure} generalizes a result of Fagin et
  al.~\cite[Theorem~5.3]{DBLP:journals/tods/FaginKRV16} on the
  expressiveness of transitive ``denial pgds.'' In our terminology,
  their theorem states that the class of \e{complete} regular spanners
  is closed under the restriction to maximal answers defined by a
  regular domination rule. \cref{thm:closure} thus extends their result
  to schemaless
  regular spanners.

Theorem~\ref{thm:closure} implies that taking the skyline relative to
regular
domination rules does not increase the expressivity of regular spanners. However, like the result of~\cite{DBLP:journals/tods/FaginKRV16},
our construction may compute VAs that are exponentially bigger than the input VA.
In Section~\ref{sct:blowups}, we will see that this is unavoidable for any
sequential VA expressing the skyline.

As an application of Theorem~\ref{thm:closure} we get in particular that regular spanners are closed under the skyline operator for most of the examples presented
earlier, i.e., Examples~\ref{exa:selfdom}--\ref{exa:leftotoright}.

\begin{toappendix}
  \subsection{Proof of Corollary~\ref{cor:closerules}}
\end{toappendix}

\begin{corollaryrep}
  \label{cor:closerules}
	There are algorithms that, given a sequential VA $P$, compute
        sequential VAs for
        $\domoper{\mathit{self}} P$,
        $\domoper{\mathit{varInc}} P$, $\domoper{\mathit{ltr}} P$, and
        $\domoper{\mathit{spanInc}}P$, respectively.
\end{corollaryrep}
\begin{proof}
  We have seen in Example~\ref{ex:varwise} that the domination relations
  $\domrela{self}$, $\domrela{varInc}$, and $\domrela{spanInc}$ can be expressed
  as variable-wise rules expressed as a single-variable domination rule defined by a regex-formula,
  and thus also by a sequential VA. We can compute which variables are used
  in~$P$, and take the Cartesian product to obtain a spanner defining the
  domination relation on the right variable sets. Note that Cartesian products
  of sequential VAs can be expressed by sequential VAs, so this gives us a
  sequential VA defining the domination
  rule that we need. We can then conclude directly with Theorem~\ref{thm:closure}.

  Notice that, in the statement, the case of $\domoper{\mathit{self}} P$ is in fact trivial because taking
  the closure under the self-domination rule has no effect so the result is
  equal to~$P$. We only state the result with
$\domoper{\mathit{self}}$
  for completeness.
\end{proof}

By contrast, Theorem~\ref{thm:closure} does not apply to the span-length
domination relation, as it is not expressible as a regular spanner
(Lemma~\ref{lem:pumping1}). In fact, we can show that taking the skyline under
this domination relation is generally \emph{not} expressible as a regular
spanner:

\begin{toappendix}
  \subsection{Proof of Proposition~\ref{prp:spanl}}
\end{toappendix}

\begin{propositionrep}
  \label{prp:spanl}
        There is a sequential VA $P$ such that 
        $\domoper{\mathit{spanLen}}P$ is not regular.
\end{propositionrep}
\begin{proof}
  We show that $\domoper{\mathit{spanLen}}P$ cannot be expressed as a VA, and
use a similar pumping argument as for Lemma~\ref{lem:pumping1}. Let $P$ be
  the spanner in the variable $x$ and the alphabet $\alphabet= \{a,b\}$ defined
  by the regex-formula:
\begin{align*}
x\{a^*\} b a^* \lor a^* b x\{a^*\}.
\end{align*}
Then on input $a^i b j^j$ with $i,j\in \mathbb{N}$, the two spans $\mspan{0}{i}$
  and $\mspan{i+1}{i+j+1}$ are extracted (before applying the skyline), and the skyline operator should remove one of these mappings (if $i \neq j$) or none (if $i = j$). We will show that
  $\domoper{\mathit{spanLen}}P$ cannot be expressed by a VA.
  By way of contradiction, assume that
  there is a VA $P'$ extracting the skyline; we assume without loss of
  generality that it is sequential. Then, on any document of
  the form $a^i b a^j$, then $P'$ extracts the span $\mspan{0}{i}$ if $i>j$ and
  $\mspan{i+1}{i+j+1}$ if $j>i$. Let $n$ be the number of states of~$P'$.
  Then on $a^{n+1} b a^{n+2}$, the VA must extract $\mspan{n+2}{2n+4}$. As in
  the proof of Lemma~\ref{lem:pumping1}, $P'$ must follow a cycle of some length
  $k$ on the run on $a^{n+1} b a^{n+2}$ that extracts this span.
  Then, going through this cycle $4$ times, $P'$ on input $a^{n+3k+1}b a^{n+2}$ extracts the span $\mspan{n+3k+2}{2n+3k+4}$. However, since $n+3k+1 > n+2$,
  this span should not have been extracted as it is dominated by $\mspan{0}{n+3k+1}$.
  Hence, we have contradicted the assumption, which concludes the proof.
\end{proof}

\subparagraph*{Other spanner formalisms.} It is natural to ask whether closure
results such as Theorem~\ref{thm:closure} also hold for other spanner
formalisms. In particular, we can ask this for the language of \emph{core
spanners}, which extend regular spanners with string equalities;
see~\cite{faginformal} for the precise definitions and \cite{schmid2020purely}
for the schemaless case. We can show that core spanners, contrary to regular
spanners, are \emph{not} closed under the skyline operator:

\begin{toappendix}
\subsection{Non-Closure of Core Spanners under Skylines}\label{app:non-closure}

  We now show that core spanners are \emph{not} closed under the skyline
  operator. To this end, let us introduce the \emph{string equality selection} operator
  used in core spanners. For any two variables~$x$ and~$y$, we
  write $\zeta^=_{x,y}$ for two variables $x$ and $y$ to denote the unary
  operator on spanners defined in the following way: given a spanner $P$ on a
  document $d$, the application $\zeta^=_{x,y} P$ of the operator to $P$ on~$d$
  captures a subset of the mappings captured by~$P$ on~$d$, consisting of those
  where one of $x$ or $y$ is undefined or they are both assigned to spans
  containing the same word. Formally, $(\zeta^=_{x,y} P)(d) = \{m \in P(d) \mid
  x \notin \dom(m) \text{~or~} y \notin \dom(m) \text{~or~} d_{m(x)} =
  d_{m(y)}\}$. As in~\cite{schmid2020purely}, we extend the notation for string equality selection as follows: for every set $X$ of variables, we let $(\zeta^=_Y P)(d) = \{m \in P(d) \mid
  \forall x,y\in Y\colon x \notin \dom(m) \text{~or~} y \notin \dom(m) \text{~or~} d_{m(x)} =
  d_{m(y)}\}$. Then, for a subset $\mathcal E = \{Y_1, \ldots, Y_r\}\subseteq \mathcal P(\SVars(P))$ we define
  $\zeta^=_{\mathcal E}(P) \colonequals \zeta^=_{Y_1}(\zeta^=_{Y_2}(\ldots \zeta^=_{Y_r}(P)\ldots ))$ (note that the order in which we apply the string equality selections does not change the outcome so this is well-defined).
  
  The way we define string string equality selections follows the usual semantics that unassigned variables correspond to
  missing information, potentially compatible with any value; see
  also~\cite{schmid2020purely} which also uses this definition for string
  equality in the schemaless case. Then \emph{core spanners} are the spanners
  that can be defined using the regular spanners as the base case and applying operators for projection, union, join and string equality selection iteratively.

The core spanners as we defined them above are schemaless, since the mappings they extract may assign to different sets of variables. In contrast, the original definition of core spanners in~\cite{faginformal} was schema-based. To make this difference explicit, let us define a notion of \emph{schema-based core spanners}:
  \begin{definition}
    \label{def:sbcore}
    A \emph{schema-based core spanner} is a core spanner in which:
\begin{itemize}
	\item the regular spanners $P$ in the definition of core spanners are all schema-based, i.e., all their output mappings always have as domain the domain $\SVars(P)$ of the spanner,
	\item for every application $\zeta^=_{\mathcal E}(P)$, we have for all $Y\in \mathcal E$ that $Y\subseteq \SVars(P)$, and
	\item whenever a union $P_1\cup P_2$ is made, we have $\SVars(P_1)\cup \SVars(P_2)$.
\end{itemize}
  \end{definition}
  An easy induction shows that schema-based core spanners in the sense of this definition indeed give core spanners that are schema-based, i.e., ensure that all captured mappings assign all the variables of the domain of the spanner. We will show the converse result later (Lemma~\ref{lem:makecoreschemabased}).  We also remark that our definition for schema-based core spanners coincides with the traditional definition of core spanners, i.e., the one in~\cite{faginformal}.

  We will show the following result:
\end{toappendix}

\begin{theoremrep}
  \label{thm:corespanclosure}
    The core spanners 
    are not closed under the skyline operator 
    with respect to the span inclusion domination relation $\domrela{spanInc}$, even
    on schema-based spanners:
    there is a schema-based core spanner $P$ such that $\domoper{\mathit{spanInc}} P$ cannot be
    expressed as a core spanner.  The same is true of the 
    skyline~$\domoper{\mathit{varInc}}$ with the variable inclusion
    domination rule. 
  \end{theoremrep}

  This result was already shown in~\cite{DBLP:journals/tods/FaginKRV16} for the
  span inclusion domination relation, but that result only showed inexpressibility as a
  schema-based core spanner. Our result extends to the schemaless setting, and
  also establishes the result for the variable inclusion domination rule.
See Appendix~\ref{app:non-closure} for the formal definitions and the proof.

We leave open the question of extending other formalisms with the skyline
operator, e.g., the 
\e{generalized core spanners} which extend core spanners 
with the difference operator~\cite{PeterfreundCFK19}, or 
the \e{context-free spanners}~\cite{DBLP:conf/icdt/Peterfreund21}
that define spanners via context-free grammars.
Note that, by contrast, closure is easily seen to hold in the formalism of \e{RGXlog}
programs, where spanners are defined using Datalog
rules~\cite{PeterfreundCFK19}. Indeed, this class consists of precisely the
polynomial-time spanners (under data complexity). Thus, for any domination
rule~$D$
for which the maximal answers can be computed in polynomial time data complexity
(in particular, for domination rules expressed as regular spanners), the result
of the skyline operator for~$D$ on an RGXlog program can be expressed as an RGXlog
program.

In the rest of this paper, we focus on applying the skyline operators
to regular spanners, with domination relations also defined via regular
domination rules.

\begin{toappendix}
  We prove Theorem~\ref{thm:corespanclosure} in the rest of this appendix. We start with some additional preparations.
 First, we will show a result that is essentially a variant of Lemma~\ref{lem:capturesubsetfunctional} for core spanners. Remember that for a spanner $P$ and a variable set $X$, the spanner $P^{[X]}$ on every document $d$ captures exactly the mappings $\{m\in P(d)\mid \dom(m)= X\}$.
 
 \begin{lemma}\label{lem:capturesubsetcore}
 	Let $P$ be a core spanner and $X$ a variable set. Then $P^{[X]}$ can be expressed as a schema-based core spanner.
 \end{lemma}
 
 For the proof we will use the following \emph{Core Simplification Lemma} from~\cite{schmid2020purely} which is a variant of an analogous result for schema-based core spanners in~\cite{faginformal}.
 
  \begin{lemma}[\cite{schmid2020purely}, Lemma~2.2]\label{lem:coresimplification}
 	For every core spanner $P$, there is a regular spanner $P_R$, $\mathcal E\subseteq \mathcal P(\SVars(P_R))$ and $Y\subseteq \SVars(P_R)$ such that $P = \pi_Y(\zeta^=_{\mathcal E}(P_R))$.
 \end{lemma}

  \begin{proof}[Proof of Lemma~\ref{lem:capturesubsetcore}]
    With the Core Simplification Lemma (Lemma~\ref{lem:coresimplification}), we may assume that $P = \pi_Y(\zeta^=_{\mathcal E}(P_R))$, where $P_R$ is regular. By Lemma~\ref{lem:decomposeregular}, we can then write $P_R = \bigcup_{X'\subseteq \SVars(P_R)} P_R^{[X']}$ where all $P_r^{[X']}$ may be assumed to be defined by functional VAs. By definition, string equality selection commutes with unions and projection commutes with union as well, so 
  		\begin{align*}
  		P &= \pi_Y\left(\zeta^=_{\mathcal E}\left(\bigcup_{X'\subseteq \SVars(P_R)} P_R^{[X']}\right)\right)\\ 
  		&= \pi_Y\left(\bigcup_{X'\subseteq \SVars(P_R)} \zeta^=_{\mathcal E}(P_R^{[X']})\right)\\ 
  		&= \bigcup_{X'\subseteq \SVars(P_R)} \pi_Y(\zeta^=_{\mathcal E}(P_R^{[X']})).
  		\end{align*}
  	For every $X'$, the spanner $\pi_Y(\zeta^=_{\mathcal E}(P_R^{[X']}))$ is schema-based, because $P_R^{[X']}$ is schema-based and applying string equality selection and projection does not change this. Moreover, $\SVars(\pi_Y(\zeta^=_{\mathcal E}(P_R^{[X']})))= Y\cap X'$. Thus, $\pi_Y(\zeta^=_{\mathcal E}(P_R^{[X']}))$ contributes tuples to $P^{[X]}$ if and only if $X'\cap Y= X$. It follows that
  	\begin{align*}
  		P^{[X]} = \bigcup_{X'\subseteq \SVars(P_R)\colon X'\cap Y = X} \pi_Y(\zeta^=_{\mathcal E}(P_R^{[X']}))
  	\end{align*}
  Now define for every $X'$ the set $\mathcal E_{X'}\colonequals \{Z\cap X'\mid Z\in \mathcal E\}$. Then, since undefined values are ignored by string equality selection and the mappings captured by $P^{[X']}$ assign exactly $X'$, we have $\zeta^=_{\mathcal E}(P_R^{[X']})= \zeta^=_{\mathcal E_{X'}}(P_R^{[X']})$, so 
  	\begin{align*}
	P^{[X]} = \bigcup_{X'\subseteq \SVars(P_R)\colon X'\cap Y = X} \pi_Y(\zeta^=_{\mathcal E_{X'}}(P_R^{[X']})).
\end{align*}
We claim that this is a representation of $P^{[X]}$ as a schema-based core spanner. Indeed, 
\begin{itemize}
	\item the $P^{[X']}$ are all schema-based by Lemma~\ref{lem:decomposeregular},
	\item for every application $\zeta^=_{\mathcal E_{X'}}(P^{[X']})$, we have for all $Z\in \mathcal E_{X'}$ that $Z\subseteq X'= \SVars(P^{[X']})$ by construction of $E_{X'}$, and
	\item in the union, all terms are over spanners with domain $X$.
\end{itemize}
Thus we have represented $P^{[X]}$ as a regular core spanner, as required. 
  \end{proof}

We will mostly be interested in the following consequence of Lemma~\ref{lem:capturesubsetcore}.
\begin{lemma}\label{lem:makecoreschemabased}
  Let $P$ be a core spanner that is schema-based. Then $P$ can be represented as a schema-based core spanner in the sense of Definition~\ref{def:sbcore}.
\end{lemma}
  What Lemma~\ref{lem:makecoreschemabased} intuitively says is that, if a core spanner is schema-based in the sense that all its captured mappings assign to exactly the same variable set, then allowing schemaless subterms in the spanner does not increase the expressivity of the model. In other words, if a schema-based spanner cannot be expressed as a core spanner in the traditional schema-based sense of~\cite{faginformal}, then it also cannot be expressed as a core spanner in the schemaless setting (i.e., with schemaless subterms, even though the overall spanner is schema-based).

  We are now ready to give the proof of
Theorem~\ref{thm:corespanclosure} for the variable inclusion domination relation:

\begin{proposition}
  \label{prp:core1}
    The core spanners
    are not closed under the skyline operator 
    with respect to variable inclusion domination relation $\domrela{varInc}$.
\end{proposition}
\begin{proof}
	We fix the alphabet $\alphabet = \{0,1\}$. We define a $0$-chunk in a word in $\alphabetstar$ to be a maximal subword consisting of only $0$, where maximality is with respect to subword inclusion. The spanner $P_{0\chunk}$ is the Boolean spanner, that is $\SVars(P_{0\chunk})=\emptyset$ that extracts the empty mapping on an input $d$ if and only if $d$ ends with a $0$-chunk that is strictly longer than all other $0$-chunks in $d$. It is known that $P_{0\chunk}$ cannot be expressed as a schema-based core spanner~\cite{faginformal} and thus, by Lemma~\ref{lem:makecoreschemabased}, it can also not be expressed as a schemaless core spanner. We will show that if core spanners were closed under skylines with respect to variable inclusion, then we could write $P_{0\chunk}$ as a core spanner, and the proposition follows directly.
	
	So let $r_1 = z\{\}\alphabetstar$, $r_2 = z\{\Sigma\}\cdot \alphabetstar$ and $r_{\mathrm{end}}= \alphabetstar\cdot x\{0^*\}$. Finally, let $\zeta^\sqsubseteq_{\{x,y\}}$ be the spanner operator such that for every spanner $P$ we have that $\zeta^\sqsubseteq_{\{x,y\}}(P)$ is a spanner capturing on every document $d$ the mappings $\{m\in P(d)\mid d_{m(x)} \text{ is a subword of } d_{m(y)}\}$. It is known that schema-based core spanners are closed under $\zeta^\sqsubseteq_{\{x,y\}}$, so in particular  $r_3 = \zeta^\sqsubseteq_{\{x,y\}}(\alphabetstar y\{\alphabetstar\}\cdot \alphabetstar \cdot x\{0^+\})$ is a core spanner.
        (Here $0^+$ abbreviates $0 0^*$.) Then we claim that 
	\begin{align}\label{eq:0chunk}
		P_{0\chunk} = \pi_\emptyset ( r_2 \Join \domoper{varInc}((\pi_x(r_3) \Join r_1))\cup r_{\mathrm{end}}).
	\end{align}
	If $d$ is the empty document, then $r_2(d)= \emptyset$, so the right-hand side does not capture any mappings. This is also true for $P_{0\chunk}$, because there are no $0$-chunks in $d$ and thus, by definition, $P_{0\chunk}$ does not capture anything. So for the empty document, (\ref{eq:0chunk}) is true.
	
        Now let $d$ be non-empty, then $r_1$ maps $z$ to $\mspan{0}{0}$ while $r_2$ maps it to $\mspan{0}{1}$. Moreover, $r_3$ maps $x,y$ to spans $m(x), m(y)$ such that $m(x)$ is a span at the end of $d$ consisting only $0$, $m(y)$ is such that $d_{m(y)}$ contains a $0$-chunk that has length at least $|m(x)|$. So $\pi_x(r_3)$ contains all suffixes of $d$ consisting only of $d$ such that there is a subword elsewhere in $d$ consisting also only of $0$ that has at least the same length. In particular, if $d$ has a $0$-chunk $c$ at the end that is longer than all other $0$-chunks in $d$, then the span of $c$ is in $r_{\mathrm{end}}(d)$ but not in $\pi_x(r_3(d))$. On the contrary, if there is no such chunk $c$, then $r_{\mathrm{end}}(d)\subseteq \pi_x(r_3)(d)$. The set $(\pi_x(r_3) \Join r_1))\cup r_{\mathrm{end}}$ contains all of $r_1(d)$ and the mappings in $(\pi_x(r_3))(d)$ to which additionally map $z$ to $\mspan{0}{0}$. When applying the skyline operator on this, the mappings in $\pi_x(r_3)(d)\cap r_{\mathrm{end}}(d)$ get eliminated. So, by what we said before, the skyline contains mappings not assigning $z$ if and only if $d$ ends with a $0$-chunk that is strictly longer than all other $0$-chunks in the document. Now, since the assignments to $x$ in $r_2$ and $r_1$ are incompatible, the join with $r_2$ leaves the resulting set of mappings non-empty if and only if there is a mapping in the skyline that does not assign $z$. So the right-hand-side of (\ref{eq:0chunk}) captures the empty mapping if and only if $d$ ends with a $0$-chunk strictly longer than any other $0$-chunk. This proves (\ref{eq:0chunk}).
	
	The spanners $r_1$, $r_2$, $r_3$ and $r_{\mathrm{end}}$ are all core spanners. Moreover, core spanners are closed under projection, join and union. So if they were closed under skylines with respect to variable inclusion, then the spanner $P_{0\chunk}$ would be a core spanner. This proves the proposition.
\end{proof}

Second, we prove the result for the span inclusion domination relation:

\begin{proposition}
  \label{prp:core2}
    The core spanners 
    are not closed under the skyline operator 
    with respect to the span inclusion domination relation $\domrela{spanInc}$.
\end{proposition}

\begin{proof}
To prove this result, let us again use the spanner $P_{0\chunk}$ from the proof of Proposition~\ref{prp:core1}.
    We will show that the spanner $r_{0\chunk}$ is in the closure of the core spanners under
    the operators $\zeta^=, \pi, {\Join}, \domoper{\mathit{spanInc}}$. This suffices to
    show the result: as core spanners are closed under string equality,
    projection, and join, if they were also closed under skyline for
    the span inclusion domination relation, then $P_{0\chunk}$ would be a core spanner. But then, by Lemma~\ref{lem:makecoreschemabased} it could also be expressed as a schema-based core spanner, which is shown to be untrue in~\cite{faginformal}.

    To express $P_{0\chunk}$, 
    let  $r_1 = \domoper{\mathit{spanInc}} \pi_{\{x\}} \zeta^=_{x, y} (\alphabetstar
    y\{0^*\}\alphabet^+ \Join \alphabetstar x\{0^*\})$. 
    (Here $\alphabet^+$ abbreviates $\alphabet \alphabet^*$.) Intuitively, $r_1$
    selects one mapping assigning $x$ to the 0-chunk at the end of the
    word which also occurs somewhere else in the word. In other
    words, if the document should be accepted by~$P_{0\chunk}$, then this will select a
    strict suffix of the longest $0$-chunk at the end of the document; otherwise
    it will select the $0$-chunk at the end of the document. To
    distinguish these cases, let
    $r_2 = \alphabetstar 0z\{0^*\}$ which selects a non-maximal suffix of the input containing only $0$, and
    let $r_3 = \pi_{\emptyset} \zeta^=_{x, z} (r_1 \Join r_2)$.
    We can now see that $r_3$ will accept the input word if and only if $r_1$
    selected a strict suffix of the $0$-chunk at the end of the document,
    i.e., if and only if the document is accepted by~$P_{0\chunk}$. Thus indeed
    $r_3$ is equivalent to $P_{0\chunk}$, and 
    by construction $r_3$ is in the closure of the core spanners under string
    equality, projection, join, and the skyline operator under the span
    inclusion domination relation. This concludes the proof.
\end{proof}

Theorem~\ref{thm:corespanclosure} follows from Proposition~\ref{prp:core1}
and~\ref{prp:core2}.
\end{toappendix}

\section{State Complexity of the Skyline Operator}\label{sct:blowups}

We have seen 
how the skyline operator does not increase
the expressive power of regular spanners, in the sense that it could be
expressed using regular operations. However, this does not 
account for the price of this transformation. In this section, we show that the size of
sequential VAs generally increases exponentially when applying the skyline operator. 
Specifically, we show the following lower bound,
for the variable
inclusion domination relation:

\begin{theorem}\label{thm:blowupskyline}
  For every $n\in \mathbb{N}$, there is a sequential VA $\calA$ with $O(n)$ states such that,
  letting~$P_\calA$ be the regular spanner that it defines, any sequential VA representing
  the regular spanner
  $\domoper{\mathit{varInc}} P_\calA$ must have $2^{\Omega(n)}$ states.
\end{theorem}

We will show in later sections how this lower bound on the state complexity of
the skyline operation can be complemented with computational complexity lower
bounds.

\subparagraph*{Proof technique: Representing Boolean functions as VAs.}
We show Theorem~\ref{thm:blowupskyline} using 
representations of Boolean
functions as sequential VAs, as we now explain.
Let $\SVars \subseteq \Vars$ be a finite set of variables (which will be used to define spanners),
and let $\vars_b\colonequals\{x_b\mid x\in \SVars\}$ be a set of Boolean variables. For
every mapping $m$ assigning spans to some of the variables in $\SVars$ (i.e.,
$\dom(m) \subseteq \SVars$),
we define a Boolean assignment $m_b\colon \vars_b \to \{0,1\}$ by setting
$m_b(x_b)\colonequals 1$ if and only if $x\in \dom(m)$, i.e., $x$
gets assigned a span by $m$. Let $P$ be a document spanner with variables
$\SVars$ and let $d$ be an input document. Then we denote by $\bool(P,d)$ the
Boolean function whose models are $\{m_b \mid m\in P(d)\}$.

Our intuitive idea is that, if the function $\bool(P,d)$ is hard to represent,
then the same should be true of the spanner $P$. To make this precise, let
us introduce the representations of Boolean functions that we work with:

\begin{definition}
A \emph{nondeterministic read-once branching program}\footnote{We remark that what we introduce here are sometimes called \emph{acyclic read-once switching and rectifier networks}, but theses are known to be equivalent to the more common definition of NROBPs up to constant factors~\cite{Razgon16}, so we do not make the difference here.} (NROBP) over the variable set $\vars_b$ is
a tuple $\Pi = (G, s, t, \mu)$
where $G = (V, E)$ is a directed acyclic graph, $s \in V$ and $t \in V$ are
respectively the \emph{source} and \emph{sink} nodes, and the function $\mu$
labels some of the edges 
with literals of variables in $\vars_b$, i.e., variables and their negations;
formally $\mu$ is a partial function from~$E$ to the literals over~$\vars_b$. We
require that, for every source-sink path $s = v_0, \ldots, v_n = t$, every
variable appears at most once in the literals labeling the edges of the path,
i.e., there are no two indices $0\leq i < j \leq n-1$ such that $\mu((v_i,
v_{i+1}))$ and $\mu((v_j, v_{j+1}))$ are both defined and map to literals of the
same variable.

  An NROBP $\Pi$ computes a Boolean function over~$\vars_b$ whose models are
  defined in the following way.
  An assignment $m_b\colon\vars_b \to \{0,1\}$ is a model of $\Pi$ if there
  is a source-sink path in~$G$ such that all literal labels on the path are
  satisfied by $m_b$, i.e., there is a sequence $s = v_0, \ldots, v_n = t$ such
  that, for each $0  \leq i < n$ for which $\ell \colonequals \mu((v_i,v_{i+1)})$ is defined, then
  the literal $\ell$ evaluates to true according to~$m_b$.
\end{definition}%

NROBPs are intuitively similar to automata. To formalize this
connection, we show how,
given a sequential VA and document, we
can efficiently compute an NROBP describing which
subsets of the variables can be assigned in captured mappings:

\begin{toappendix}
  \subsection{Proof of Lemma~\ref{lem:torobp}}
\end{toappendix}

\begin{lemmarep}\label{lem:torobp}
 Let $P$ be a regular spanner on variable set $\SVars$ represented by a sequential VA~$\mathcal A$ with $n$ states. Then, for every document $d$, there is an NROBP $G$ representing $\bool(P,d)$ with $O(|d| \times n \times |\SVars|)$ nodes.
\end{lemmarep}

\begin{proofsketch}
  We intuitively compute the product of the VA with the input document, to
  obtain a directed acyclic graph representing the runs of the VA on the
  document. We obtain the NROBP by relabeling the marker transitions and
  performing some other modifications.
\end{proofsketch}

\begin{proof}
  In this proof, we use VAs extended with $\epsilon$-transitions (see Appendix~\ref{app:prelim}).
  We further assume that the VA is \emph{$\varepsilon$-cycle-free}, namely, that
  there are no cycles or self-loops of $\varepsilon$-transitions. This condition
  can be enforced on the automaton in linear time. Specifically, considering the
  $\varepsilon$-transitions, compute the strongly connected components of the
  resulting directed graph. Now, for every strongly connected component (set of
  states), merge the corresponding states into a single state: make it initial
  if one of the merged states was the initial state, make it final if one of the
  merged states was final, and add transitions from the new state to reflect all
  transitions that can be performed from one of the merged states.
  Finally, remove all self-loops of the
  $\varepsilon$-transitions.
  Clearly this
  transformation does not affect the semantics of the automaton, because any
  accepting run in the original automaton yields one in the new automaton that
  goes via the merged states, does not need to use the removed
  $\varepsilon$-transitions, and reads the same letters and markers in the same
  order; conversely any accepting run of the new automaton can be rewritten
  by navigating using the removed $\varepsilon$-transitions. Note that this
  transformation also does not affect sequentiality of the automaton, and can only make the number of states decrease.

  Thanks to the addition of $\varepsilon$-transitions, we can assume without loss
  of generality that $\mathcal A$ has exactly one final state, and that this
  state has no outgoing transitions. Indeed, we can enforce it in linear time
  simply by adding a new final state which we can reach from every old final
  state by an $\varepsilon$-edge, and by making all old final states non-final.
  Clearly, this does not change the extracted spans, nor does it affect sequentiality or the previous transformations. Further, the number of states only increases by~$1$: this will not be a problem towards establishing the size bound.

  Last, we assume that the automaton is trimmed (see Appendix~\ref{app:prelim}).

  In the remainder, let $s$ be the initial state and let $t$ be the final state of $\mathcal A$. %
 We construct $G$ in several steps. The purpose of the first step is to get rid of the input word $d$ by essentially constructing the product of $\mathcal A$ and $d$: for every state $x$ of $G$ and for every prefix $p$ of $d$, we create a new state $(x,p)$. We connect these states as follows: if $x$ and $y$ are states in $\mathcal A$ such that there is an $\varepsilon$-edge from $x$ to $y$, we add for every prefix $p$ of $d$ the $\varepsilon$ edge $(x, p)(y,p)$. If there is an edge $xy$ in $\mathcal A$ on which a letter $a$ is read, then, for every prefix $p$ such that $p\cdot a$ is also a prefix of $d$, we add the edge $(x,p)(y, pa)$. In both cases, if there is a variable marker on the edge $xy$, we also add that operation on the new edges. Call the resulting digraph $G'$. Clearly, $G'$ has at most $|d|n$ vertices. Let $s'\colonequals (s, d)$ and $t'=(t, \varepsilon)$.
  Then there is a $s'$-$t'$-path in $G'$ with variable markers on a variable set $X\subseteq \SVars$ if and only if there is an accepting run of $\mathcal A$ that assigns spans to the variables in $X$.
 
  We now observe that, because the initial VA is $\varepsilon$-cycle-free, then the graph $G'$ is acyclic. Indeed, assume by contradiction that there is a cycle in~$G'$. By construction this cycle must correspond to a cycle in the VA. The cycle in~$G'$ must only involve vertices whose second component is a given prefix $p$ of~$d$, because there are no edges of~$G'$ going from a longer to a shorter prefix. Hence, the cycle in the VA cannot include any letter transitions. As the VA is sequential and trimmed, we know the cycle cannot involve a marker, as otherwise we can build from the cycle an accepting run where the same marker is assigned twice. Thus, the cycle only consists of $\varepsilon$-transitions, but this is a contradiction because the VA is cycle-free.

 We now turn $G'$ into a NROBP by substituting all opening variable markers $\vdash x$ for each variable $x\in \SVars$ by the literal $x_b$ and deleting all all other markers. Call the result $G''$. Again, since $\mathcal A$ is sequential and trimmed, for every $s'$-$t'$-path in $G'$ every marker $\vdash x$ is followed by a marker $\dashv x$ on the path. It follows that that $G''$ accepts exactly the assignments $m_b$ for which there is an accepting run of $\mathcal A$ on $d$ such that for all variables $x$ that are assigned a span, we have $m_b(x_b) = 1$. Note however that there might be variables $x_b$ such that $m_b(x_b)=1$ but $x$ is not assigned a span on the run. To prevent this from happening, we have to force all variables that are not seen on a path in $G''$ to $0$ in the corresponding model. To do so, we use a variant of the usual completion technique for ROBPs: for every node $u$, we compute a set $V(u)$ that contains all variables that appear on edges on paths from the source $s$ to $v$. Now iteratively, for every edge $uv$, when there is a variable $y\in V(v)\setminus V(u)$ that is not the label of $uv$, we substitute $uv$ by a path $uv_yv$ of length two where $v_y$ is a new edge, give $uv_y$ the label of $uv$ if it has one, and give $v_yv$ the label $\neg y$. Doing this exhaustively yields a new NROBP $G$ in which on every $s'$-$t'$-path for every variable $x\in \SVars$ there is exactly one edge having as edge label either $x_b$ or $\neg x_b$. Moreover, there is a bijection $\pi$ between $s'$-$t'$-paths in $G''$ and $G$ such that for every path $p$ in $G''$ we have: all variables that appear on $p$, also appear on $\Pi(p)$ and they appear there positively. All other variables appear negatively. It follows that $G$ computes $\bool(P, d)$. Observing that every edge in $G''$ gives rise to a path of length at most $|\SVars|+1$ establishes the bound on size and completes the proof.
\end{proof}

We will now use the fact that NROBPs are exponentially less concise than other
Boolean function representations. 
Namely, we define a \emph{read-3 monotone 2-CNF formula} on a set of
variables~$X$ as a conjunction of clauses which are disjunctions of 2 variables
from~$X$, where each variable appears at most~$3$ times overall. 
We use the fact that converting such formulas to NROBPs can incur an
exponential blowup. This result is known (see, e.g.,~\cite{BovaCMS14}) but we
give a proof in Appendix~\ref{apx:cnflb} for convenience:

\begin{toappendix}
  \subsection{Proof of Proposition~\ref{prp:cnflb}}
  \label{apx:cnflb}
\end{toappendix}

\begin{propositionrep}[\cite{BovaCMS14}]
  \label{prp:cnflb}
  For any $n \in \mathbb{N}$, there is a read-3 monotone 2-CNF formula $\Phi$
  on~$n$ variables having
  size $O(n)$ such that 
  every representation of $\Phi$ as an NROBP has size $2^{\Omega(n)}$.
\end{propositionrep}

\begin{proof}
  Given a graph $G=(V,E)$ of maximal degree~3, we construct a read-3 monotone 2-CNF $F_G$ in the variables $\{x_v\mid v\in V\}$ as
  \begin{align*}
  	F_G\colonequals \bigwedge_{uv\in E} x_u\lor x_v.
  \end{align*} 
In~\cite[Theorem~8.5]{AmarilliCMS20}, it is shown that any NROBP computing $F_G$ has size $2^{\Omega(\mathrm{tw}(G))}$ if the maximum degree of $G$ is bounded by a constant. Here $\mathrm{tw}(G)$ denotes the treewidth of $G$.
Using the fact that there exist graphs of maximal degree~3, size $O(n)$, and treewidth
  $\Omega(n)$~\cite[Proposition~1, Theorem~5]{GroheM09} yields the claim.
\end{proof}

We now conclude the proof of Theorem~\ref{thm:blowupskyline}, sketched below (see
Appendix~\ref{apx:blowups} for details):

\begin{proofsketch}%
  Given a read-3 monotone 2-CNF formula $\Phi$,
  we show how to construct a regular spanner on which the skyline operator
  captures mappings corresponding precisely to the satisfying assignments of~$\Phi$. As
  a sequential VA expressing this spanner can be efficiently converted to an NROBP by
  Lemma~\ref{lem:torobp}, we can conclude that, when applied to the family of
  formulas from Proposition~\ref{prp:cnflb}, all sequential VA representations have
  exponential size.
\end{proofsketch}

\subparagraph*{An independent result: Lower bound on the state complexity of
schema-less joins.}
We believe that the connection to Boolean functions used to show
Theorem~\ref{thm:blowupskyline} can 
be of independent interest as a general technique to show lower
bound on the state complexity of document spanners. 
Indeed, independently from the skyline operator, we can also use this connection 
to show a lower bound on the size of sequential VAs representing the
\emph{natural join} of two
regex-formulas. The \emph{natural join operator} is a standard operator on spanners that
merges together compatible mappings: see Appendix~\ref{app:prelim} for the
formal definition. We have:

\begin{theorem}\label{thm:blowupalgebraic}
	For every $n\in \mathbb{N}$, there are regex-formulas $e_n$ and $e_n'$
        of size $O(n)$ such that every sequential VA equivalent to $e_n \bowtie e_n'$ 
        has $2^{\Omega(n)}$ states. 
\end{theorem}

This result is the counterpart for state
complexity of the \NP-hardness 
of evaluating the join of two
regex-formulas~\cite{PeterfreundFKK19}. It only holds in the schemaless case;
indeed
in the schema-based case it is known that the join of two
functional VAs can be computed as a functional VA
in polynomial
time~\cite{DBLP:conf/pods/FreydenbergerKP18}.

\begin{toappendix}
  \label{apx:blowups}
\subsection{Proof of Theorem~\ref{thm:blowupalgebraic}}\label{sct:blowupalgebraic}
  We first show in this section the proof of Theorem~\ref{thm:blowupalgebraic}, before showing in the next section the proof of Theorem~\ref{thm:blowupskyline}.

The proof is inspired by the \NP-hardness proof for evaluation of spanners given as the join of two regex formulas~\cite[Theorem~3.1]{PeterfreundFKK19}. We define two regexes $r_c$ and $r_=$ whose join will encode a CNF-formula. So let $F= C_1\land \ldots \land C_m$ be a 2-CNF-formula in variables $x_1, \ldots x_n$. We 
use the formulas from Proposition~\ref{prp:cnflb}. We remind the reader that those formulas are monotone $2$-CNF, i.e., all clauses contain only two variables and all variables appear only positively. Moreover, every variable appears in only three clauses. We first define $r_=$ as
\begin{align*}
 r_= \colonequals r_{x_1}\cdot \cdots \cdot r_{x_n}\cdot a
\end{align*}
where $a$ is a letter and 
\begin{align*}
 r_{x_i} \colonequals \left(x_i^1\{\varepsilon\}x_i^2\{\varepsilon\}x_i^3\{\varepsilon\}\right) \lor \varepsilon.
\end{align*}
Note that, considering a successful evaluation of $r_=$ on the single letter
  document $a$, for each $i\in [n]$, either the evaluation maps 
the three variables $x_i^1, x_i^2, x_i^3$ to the span $[0,0\rangle$,
  or it does not map any of these three variables. The fact that the three
  variables have to be matched in the same way 
  will allow us to simulate variable assignments to $x_i$ consistently later on.

Next we construct $r_C$ by setting 
\begin{align*}
 r_C = a \cdot r_{C_1}\cdot \ldots \cdot r_{C_m}
\end{align*}
where
\begin{align*}
 r_{C_j} \colonequals x_{j_1}^{p(j,j_1)}\{\varepsilon\} \lor x_{j_2}^{p(j,j_2)}\{\varepsilon\}
\end{align*}
where $x_{j_1}$ and $x_{j_2}$ are the two variables in $C_j$ and $p(j,j_2)$ is the number $s\in \{1,2,3\}$ 
such that the appearance of $x_j$ in $C_j$ is the $s$-th appearance of $x_j$ in $F$ when reading from left to right. Note that this is well-defined because every variable appears in at most three clauses.

Now consider the results of the spanner $r \colonequals r_=\bowtie r_C$ on the input
  $d=a$. A mapping $m$ is in $r(a)$ if and only if for every clause $C_j$ of $F$
  there is a variable $x_b$ in $C_j$ such that $x_b^{p(j,b)}$ is not mapped to $[0,0\rangle$ by $m$.

Let $\mathcal A$ be a sequential variable-set automaton for $r$ which we assume
  without loss of generality to be trimmed, and consider
  again the input $d=a$. Then, for every state $s$ of $\mathcal A$ we can
  determine if in an accepting run of $\mathcal A$ the state $s$ can be reached
  before the letter $a$ is read or after (note that both cannot happen for the
  same state $s$ in different runs: as the VA is sequential and trimmed, it would imply that the automaton can accept a document with more than one occurrence of~$a$, which is not allowed by~$r$).
  Now for all edges going from states in which $a$ has been read before, delete all variable markers. By what was said before, the resulting automaton $\mathcal A$ accepts exactly the mappings $m$ for which there is a satisfying assignment $m_b$ of $F$ such that $m$ assigns the variables $\{x_i^j\mid m_b(x_i) = 0\}$. Now deleting all variable markers in variables $x_i^2$ and $x_i^3$ and substituting all $x_i^1$ by $x_i$ yields a VA $\calA'$ representing a spanner $P'$ such that $\bool(P',a)= \{1-m_b\mid m_b \text{ is a model of } F\}$, i.e., the models of $F$ up to flipping all bits of the satisfying assignments. Note that this VA $\calA'$ has size linear in~$\calA$.

Now applying Lemma~\ref{lem:torobp} to~$\calA'$ we get an NROBP $G$ computing $\bool(P', a)$ and thus, by flipping the sign of all variables in $G$, we get an NROBP $G'$ computing $F$. Moreover, the size of $G'$ is $O(|\mathcal A| n)$. But by the choice of the CNF $F$, we know from Proposition~\ref{prp:cnflb} that any NROBP encoding $F$ has size $2^{\Omega(n)}$ which completes the proof.

  \subsection{Proof of Theorem~\ref{thm:blowupskyline}}\label{sct:blowupskyline}
  We next show the proof of Theorem~\ref{thm:blowupskyline}, which establishes that the skyline operator under the variable inclusion domination rule generally causes an exponential blowup. This result can be intuitively understood as a variant of proof establishing that evaluating this operator is 
  \NP-hard in combined complexity (Theorem~\ref{thm:nphardspanner}), shown later in the paper; the proof also uses similar ideas to the proof of Theorem~\ref{thm:blowupalgebraic} shown earlier.

We again fix a monotone 2-CNF formula $F$ in $n$ variables $\{x_1, \ldots, x_n\}$ in which 
  we assume that every variable appears in only three clauses. Let $C_1,
  \ldots, C_m$ be the clauses of $F$. We choose $F$ such that any
  NROBP for $F$ has size $2^{\Omega(n)}$, which is possible according to Proposition~\ref{prp:cnflb}.

Now consider two spanners defined by slightly modifying the two spanners in the
proof of Theorem~\ref{thm:nphardspanner}. To this end, let $T_i$ be defined as
the indices of the clauses in which $x_i$ appear; remember that $F$ is monotone so variables always appear positively. Moreover, introduce for every
$i\in [n]$ three spanner variables $v_{i,j}$, one for each $j\in T_i$, and an
additional variable $\bar x_i$. Then we set:
\begin{align*}
 r_\valid \colonequals \cdot_{i\in [n]} ((x_i\{\varepsilon \} \cdot \cdot_{j \in T_i} v_{i,j}\{\varepsilon\})\lor \bar{x}_i\{\varepsilon\}).
\end{align*}
Intuitively, in every match, the spanner assigns to some of the $x_i$ the span $[0,0\rangle$, which we interpret as setting the variable $x_i$ to true in the CNF. Moreover, the clauses that are made true by setting $x_i$ to true are set to true as well by assigning them $[0,0\rangle$. Finally, the assignments to $\bar x_i$ encode the negation of the assignments to the $x_i$.

For $r_\mask$ we set:
\begin{align*}
r_{\mask} = \bigvee_{k \in [m]} (\conR_{i\in [n]} (x_i\{\bar{x}_i\{\varepsilon\}\}( \conR_{j \in [T_i]\setminus \{k\}} v_{i,j}\{\varepsilon\})).
\end{align*}
Note that $r_\mask$ on input $\emptyword$ matches in exactly $m$ ways: in each
of these matches, all variables $x_i, \bar x_i$ are assigned $[0,0\rangle$.
Moreover, in the $k$-th match all of the $v_{i,j}$ except those of the form $v_{i,k}$ are assigned.

Then the complete regular spanner is $r\colonequals r_\mask \lor r_\valid$, and we apply the skyline
operator $\domoper{\mathit{varInc}}$ on $r$. We consider the evaluation of this spanner on the empty document $\emptyword$.

Let us first study which domination relation can hold on the mappings captured by this spanner. Note first that the matches of $r_\valid$ do not dominate each other, since in each of them for every $i\in [n]$ either $x_i$ or $\bar x_i$ is assigned and all those assignments differ. Similarly, the matches of $r_\mask$ cannot dominate each other since they are not assigning the $v_{i,k}$ for different values of $k$ and thus are pairwise incomparable. Finally, since all matches of $r_\mask$ assign all $x_i$ and $\bar x_i$, none of them is dominated by any match of $r_\valid$. So the only domination that may happen is that matches of $r_\mask$ can dominate matches of $r_\valid$. Moreover, this happens exactly for the matchings $m$ of $r_\valid$ for which there is an index $k\in[m]$ for which no variable of the form~$v_{i,k}$ is assigned, i.e., there is a clause $C_k= x_{i_1}\lor x_{i_2}$ of $F$ such that in the corresponding assignment
none of the two variables $x_{i_1}$ and $x_{i_2}$ are made true. In other words, a mapping of $r_\valid$ is dominated by a mapping of~$r_\mask$ exactly when the valuation of the $x_i$'s that it describes does \emph{not} satisfy the CNF~$F$.

Now consider the function $f$ that is $\bool(\domoper{\mathit{varInc}} r, \emptyword)$ projected to the variables $x_1, \ldots, x_n$. This function describes which subsets of~$x_1, \ldots, x_n$ are assigned a span (necessarily~$[0,0\rangle$) in the mapping. We claim these are precisely the assignments satisfying the CNF~$F$. Indeed, the captured assignments include the all-$1$ assignment (captured by $r_\mask$), which is indeed a satisfying assignment of~$F$: and it includes all assignments of $r_\valid$ that are not dominated by an assignment of~$r_\mask$, those which satisfy~$F$ by our preceding discussion. Conversely, for every satisfying assignment of~$F$ we obtain a mapping of~$r_\valid$ which is not dominated by a mapping of~$r_\mask$ and is part of the output (note that the all-$1$ assignment is captured both thanks to $r_\mask$ and thanks to $r_\valid$).
Thus, indeed the function $f$ is in fact exactly~$F$.

To conclude, fix a variable-set automaton $\mathcal A$ which represents $\domoper{\mathit{varInc}} r$.
Then using Lemma~\ref{lem:torobp}
and the fact that projection does not increase the size
of NROBP, we get an NROBP $G$ for $F$ of size $O(|\mathcal A|n)$. Since we have
chosen $F$ such that the size of the NROBP $G$ is at least $2^{\Omega(n))}$, it follows that $|\mathcal A| = 2^{\Omega(n)}$ which completes the proof of Theorem~\ref{thm:blowupskyline}.
\end{toappendix}

\section{Complexity of the Skyline Operator}
\label{sec:complexity}

We have shown that the skyline operator applied to regular spanners cannot be
expressed as a regular spanner without an exponential blowup in the size, even
for domination rules expressed as regular spanners (namely, for the variable
inclusion domination rule). We now study whether we can efficiently evaluate the skyline
operator without compiling it into the automaton. Formally, we study its computational
complexity of skyline extraction:

\begin{definition}
  \label{def:problem}
  The \emph{skyline extraction problem} is the following: given a document~$d$, a sequential VA~$\calA$ capturing a regular spanner $P_\calA$, and a domination rule $\domrule{}$ expressed as a
  sequential~VA, compute the set of mappings in the results of the skyline operator $(\domoper{\domrule{}} P_\calA)(d)$.
\end{definition}

\subparagraph*{Data complexity.}
We start by observing that skyline extraction is clearly tractable in the data
complexity perspective in which $d$ is the only input:

\begin{proposition}
  For any fixed sequential VA $\calA$ and domination rule expressed as a sequential VA~$\domrule{}$, 
  the skyline extraction problem for~$P_\calA$ and~$\domrule{}$ can be solved in
  polynomial time data complexity, i.e., in polynomial time in the input~$d$.
\end{proposition}

\begin{proof}
  We simply materialize the set of all captured mappings $(P_\calA)(d)$, in polynomial time because $\calA$ is fixed. Then,
  for any pair of mappings, we can check if the domination relation holds
  using the domination rule $\domrule{}$; this is again in polynomial time. 
  We then return the set of maximal mappings in polynomial time.
\end{proof}
Note that this result would easily extend to fixed expressions using multiple skyline operators together with regular spanner operators, as all these operators are polynomial-time.

\subparagraph*{Combined complexity.}
We now turn to combined complexity settings in which the domination rule $\domrule{}$ or the spanner $P$ are considered as part of the input. Remember that we focus on regular spanners represented as sequential VAs, since for those it is known that the combined complexity of spanner evaluation is output polynomial~\cite{DBLP:conf/pods/MaturanaRV18}.

As we have seen in Section~\ref{sct:closure}, in terms of expressiveness, the regular spanners are closed under all domination rules expressible as regular spanners, in particular those of Examples~\mbox{\ref{exa:selfdom}--\ref{exa:leftotoright}}.
However, we have seen in Section~\ref{sct:blowups}
that compiling the skyline into the VA may generally incur an exponential
blowup, already for fixed domination rules. 
This bars any hope of showing tractability of the skyline extraction problem by applying 
known evaluation algorithms on the result of this transformation (e.g.,~those
from~\cite{DBLP:conf/pods/FreydenbergerKP18,DBLP:conf/pods/FlorenzanoRUVV18,constdelay}),

This leads to the question if there are other approaches to solve the skyline
extraction problem with efficient combined complexity, without materializing an equivalent
VA. In this section, we show that this is not the case, assuming
$\mathsf{P}
\neq \mathsf{NP}$. Our lower bound already holds for a fixed domination rule, namely, the variable inclusion domination rule; and in fact it even holds
in \emph{query complexity}, i.e., when
the document is fixed.

\begin{theorem}\label{thm:nphardspanner}
  There is a fixed document $d$ such that the following problem is \NP-hard:
  given a sequential VA $\calA$ encoding a regular spanner~$P_\calA$ and a number~$n \in \mathbb{N}$,
  decide whether 
  $(\domoper{\mathit{varInc}}P_\calA)(d)$
  contains more than~$n$ mappings.
\end{theorem}

This will imply that, conditionally, the skyline extraction problem is intractable 
in combined complexity.
We state this in the language of \emph{output-polynomial algorithms}, where an
algorithm for a problem $f\colon\alphabetstar\rightarrow\alphabetstar$ runs in
\emph{output-polynomial time} if, given an input~$x$,
it runs in time polynomial in $|x| + |f(x)|$.
Namely, we use the following folklore connection
between 
output-polynomial time and decision problems, see
e.g.~\cite{CapelliS19} for a similar construction:

\begin{lemmarep}\label{lem:outputpoly}
	Let $f:\alphabetstar \rightarrow \alphabetstar$ and let $p$ be a
        polynomial. Assume that it is $\mathsf{NP}$-hard, given an input $x$
        and integer $k\le
        p(|x|)$, to decide if $|f(x)|< k$.
        Then there is no output polynomial time algorithm for $f$, unless $\mathsf{P} = \mathsf{NP}$.
\end{lemmarep}
\begin{proof}
By way of contradiction, assume that there is an output polynomial time algorithm for $f$ that computes $f(x)$ in time $r(|x|+|f(x)|)$ for some polynomial $r$. Assume w.l.o.g.~that~$r$ is monotone. We show that there is then a polynomial time algorithm for the decision problem, and the lemma follows directly.

The algorithm works as follows: given $x$ and $k$, we simulate the output polynomial time algorithm for $r(|x|+ k+1)$ steps. If the simulation terminates in that many steps, we simply check whether $|f(x)| \geq k$. If the simulation does not terminate, we reject the input. 

We claim that the above is a polynomial time algorithm for the decision problem in the statement of the lemma. First note that for the runtime, it suffices to show that we simulate the output polynomial time algorithm only for a polynomial number of steps. But this is true because $r(|x|+k+1) \le r(|x|+ p(|x|)+1)$ is polynomially bounded in $|x|$.

It remains to show correctness. First, if the simulation terminates in the allowed number of steps, the output is obviously correct. So assume that the simulation does not terminate. If $f(|x|)< k$, then the output polynomial time algorithm terminates by assumption in $r(|x|+|f(x)|) \le r(|x|+k)$ steps, so the simulation must terminate. So from the fact that it does not terminate, we can infer that $|f(x)|\ge k$ and thus it is correct to reject.
\end{proof}

From Theorem~\ref{thm:nphardspanner} and Lemma~\ref{lem:outputpoly}, we directly get our intractability result:

\begin{corollary}
  \label{cor:varinc}
    Unless $\mathsf{P} = \mathsf{NP}$,
    there is no algorithm for the skyline extraction problem with respect to the
    variable inclusion domination rule that is
    output-polynomial in combined complexity (i.e., in the input sequential VA), even when the input document is fixed.
\end{corollary}
Note that this result is incomparable to Theorem~\ref{thm:blowupskyline}: lower
bounds on the size of equivalent VAs generally do not preclude the existence of
other algorithms that are tractable in combined complexity, and conversely it
could in principle be the case that evaluation is intractable in combined
complexity but that there are small equivalent VAs that are intractable to
compute.
Besides, the proofs are also different. Namely, the proof of Theorem~\ref{thm:blowupskyline} used monotone 2-CNF formulas, for which we could compute spanners giving an exact representation of the satisfying assignments, but for which the satisfiability problem is tractable. As we will see, the proof of Theorem~\ref{thm:nphardspanner} uses the intractability of SAT on CNF formulas, but does not use an exact representation of the satisfying assignments.

\subparagraph*{Proving Theorem~\ref{thm:nphardspanner}.}
We give the proof of Theorem~\ref{thm:nphardspanner} in the rest of this section, together with an additional observation at the end. In the next
section, we will study how hardness can be generalized to other domination
relations (in particular all domination relations introduced in
Section~\ref{sec:skyline}
except the trivial self-domination relation), and will investigate the existence of
tractable cases.

\begin{inlineproof}[Proof of Theorem~\ref{thm:nphardspanner}]
    We reduce from the satisfiability problem SAT.
    Let $F$ be a CNF formula with $n_x$ Boolean variables $x_i$ with $i \in
    [n_x]$ and $n_c$ clauses $C_j$ with $j \in [n_c]$.
    For convenience, define the set $T_i = \{j \mid x_i \text{ appears positively in } C_j\}$,
    and define the set $F_i = \{j \mid x_i \text{ appears negatively in } C_j\}$.
    We will build a regular spanner on variables $v_{i,j}$ for $i\in [n_x]$ and $j\in [n_c]$, together with
    a special variable~$a$.

    We will define two spanners $r_{\valid}$ and $r_{\mask}$, both as regex
    formulas, and will evaluate them on the empty document $d = \epsilon$. 
    Let us first sketch the idea:
    the spanner $r_{\valid}$ will extract 
    one mapping for each possible assignment to the variables of $F$. 
    Each such mapping will encode which clauses get satisfied by which variable in the assignment,
    by assigning spans to the corresponding spanner variables $v_{i,j}$.
    The second spanner $r_{\mask}$ will capture $n_c$ additional mappings which
    will be maximal (thanks to the additional variable~$a$) and will each dominate 
    the mappings captured by $r_{\valid}$ for which the corresponding
    assignment     does \emph{not} satisfy a specific clause of~$F$.
    This will ensure that $F$ is satisfiable if and only if 
    there are strictly more than $n_c$ mappings in the skyline of $r_{\valid}
    \lor r_{\mask}$ on~$d$.
    
    Formally, we define the spanners as regex-formulas, where the dots denote
    concatenation:
    \[
    	\!\!\!\!\!r_{\valid} = \conR_{i \in [n_x]}((\conR_{j \in T_i} v_{i,j}\{\epsilon\}) \lor
        (\conR_{j \in F_i} v_{i,j}\{\epsilon\})) \hfill
        r_{\mask} = a\{\epsilon\}\conR\bigvee_{k \in [n_c]} \conR_{i \in [n_x], j
        \in[n_c]\setminus \{k\}} v_{i,j}\{\epsilon\}.
      \]
    This definition is in polynomial time in the input CNF~$F$.

    Note that the mappings captured by $r_\mask$ are never dominated. First, they do not dominate each other: each of them assigns no $v_{i,k}$ for some $k$. Further, all mappings of $r_\mask$ assign $a$ and all mappings of $r_\valid$ do not, so the latter cannot dominate the former.

    To construct a CNF variable assignment 
    from a mapping $m$ captured by $r_{\valid}$,
    we use the following encoding:
    if the mapping $m$ assigns the span $[0, 0\rangle$ 
    to the spanner variable $v_{i,j}$ then this encodes that 
    the variable $x_i$ appears in the clause $C_j$ 
    and $x_i$ is assigned in a way that satisfies~$C_j$.
    The definition of $r_{\valid}$ ensures that all variables appearing at
    least once will be assigned exactly one truth value among true or false.

        We claim that on $d = \epsilon$, the skyline
        $(\domoper{\mathit{varInc}}(r_\valid \lor r_\mask))(d)$ contains at least $n_c+1$ mappings if and only if $F$ is satisfiable.
    Assume first that $F$ is satisfiable, and let
    $v$ be a satisfying assignment.
    Then there is a corresponding mapping $m$ captured by $r_{\valid}$ encoding~$v$.
    Indeed, as $v$ satisfies all clauses,
    for every clause index $j \in [n_c]$ 
    there is a variable $x_i$ assigned by~$v$ in a way that makes~$C_j$ true,
    i.e., $v_{i,j}$ is assigned.
    Hence $m$ will not be dominated by any mapping captured by $r_{\mask}$. 
    Thus, the skyline of $r_{\valid} \lor r_{\mask}$ must contain some mapping
    captured by~$r_\valid$, namely, either $m$ or some other mapping captured by~$r_\valid$ which dominates~$m$. In all cases, the skyline must have at least $n_c + 1$ mappings.

    Now assume the skyline of $r_{\valid} \lor r_{\mask}$ has at least $n_c+1$ mappings.
    By construction, $r_{\mask}$ captures exactly $n_c$ maximal mappings, so
    there is at least one mapping $m$ in the skyline which is captured by $r_{\valid}$. This mapping $m$ encodes an assignment $v$ of the variables of~$F$. 
    As $m$ is not dominated by any mapping captured by $r_{\mask}$,
    for each clause index $j \in [n_c]$ 
    there must exist a variable index $i \in [n_x]$ 
    such that $v_{i,j}$ is assigned.
    Therefore $v$ is a satisfying assignment of~$F$.
    Overall, we have shown that $F$ is satisfiable if and only if
    $\domoper{\mathit{varInc}}(r_{\valid} \lor
    r_{\mask})$ has at least $n_c+1$ satisfying mappings, which concludes the proof.
    \end{inlineproof}

    We last notice that we can modify Corollary~\ref{cor:varinc} slightly:
    instead of applying to a fixed single-variable domination rule
    that implicitly represents a product, the result also applies when the domination rule is specified explicitly on the entire domain as a regular spanner:

\begin{corollaryrep}
	Assuming $\mathsf{P}\ne \mathsf{NP}$, there is no algorithm for the
        skyline extraction problem which is output polynomial in combined
        complexity even if the domination relation is given as one
        sequential VA 
        (not by implicitly taking the product of
        single-variable sequential VAs).
\end{corollaryrep}

\begin{proof}
  By inspection of the proof of Theorem~\ref{thm:nphardspanner}, we notice that
  hardness already holds for the variable inclusion domination relation on the
  empty document. Fortunately, we can express a domination rule that expresses
  this domination relation on the empty document in polynomial time in the
  domain $X$ of variables. Intuitively, the captured mappings should be those
  where we assign a certain subset $X'\subseteq X$ of the left-hand-side
  variables and assign a superset of $X'$ as the right-hand-side variables.
  Formally, we can take the domination rule:
  \[
    \conR_{x \in X} \left(\epsilon \vee \dagg{x}\{\epsilon\} \vee
    x\{\dagg{x}\{\epsilon\}\}\right)
  \]
  (Notice that this only works on the empty document, and crucially relies on
  the fact that the markers at every position can be ordered in an arbitrary way
  provided that for each variable the opening marker comes before the closing
  marker. If the document were
  non-empty, then we could not perform a concatenation across the successive
  variables.)

  Thus, given a CNF formula, we build a spanner~$P$ as in the proof of
  Theorem~\ref{thm:nphardspanner}, compute the domination rule on the requisite
  set of variables in polynomial time as expressed above, and conclude like in
  the earlier proof and using Lemma~\ref{lem:outputpoly}.
\end{proof}

\section{Intractable and Tractable Domination Rules}
\label{sec:further}

We have shown that the skyline extraction problem is intractable in combined
complexity for regular spanners, and this intractability already holds
in the case of a fixed variable-wise domination rule, namely, the variable inclusion rule. However, this leaves open
the same question for other domination rules, e.g., for the span inclusion rule
-- in particular if we restrict our attention to schema-based spanners, which are
typically better-behaved (e.g., for the complexity of the join and
difference operators~\cite{faginformal}).

In this section, we show that, unfortunately, hardness still holds in that
context. Specifically, we introduce a condition on domination rules, called having
\emph{unboundedly many disjoint strict domination pairs} (UMDSDP).
This condition is clearly satisfied
by our example domination rules (except self-domination). We then show that
UMDSDP is a sufficient condition for intractability: this result re-captures the hardness
of variable inclusion (Theorem~\ref{thm:nphardspanner}) and also shows hardness
for the span inclusion, left-to-right, and span length domination rules.

We then introduce a
restricted class of domination rules, called \emph{variable inclusion-like}, and
show that on this class a variant of the UMDSDP condition in fact
\emph{characterizes} the intractable
cases. In particular, all such domination rules without the condition
enjoy tractable skyline extraction. Last, we study additional examples for
general domination rules, and show that among rules not covered 
by UMDSDP, some are easy and some are hard.

\subparagraph*{The UMDSDP condition.}
To introduce our sufficient condition for intractability of skyline extraction, 
we first define 
\emph{disjoint strict domination pairs}.

\begin{definition}
    For a domination relation $\domrela{}$ on a document~$d$,
    a \emph{strict domination pair} of~$\domrela{}$ on~$d$
    is a pair $(s_1, s_2)$ of spans
    with $s_1 \domrela{} s_2$ and $s_1 \ne s_2$.
    Two such pairs 
  $\{s_1, s_2\}$ and $\{s_1', s_2'\}$
  are \emph{disjoint} if, letting $s$ be the smallest span containing $s_1$ and
  $s_2$, and letting $s'$ be the smallest span containing 
  $s_1'$ and $s_2'$,
  then $s$ and $s'$ are disjoint. Otherwise, they \emph{overlap}.
\end{definition}

When in a domination pair $(s_1, s_2)$ one of the $s_i$ is not assigned, we write this 
with a dash: ``$-$''. For the purpose of disjointness, we say that ``$-$'' is contained in every span.

\begin{example}
  The pairs $(\mspan{1}{3}, \mspan{2}{4})$ and $(\mspan{9}{10},
  \mspan{6}{8})$ are disjoint. The pairs $(\mspan{1}{3}, \mspan{7}{9})$ and
  $(\mspan{4}{6}, \mspan{10}{12})$ overlap (even though all of the constituent
  spans are disjoint). Finally, $(-, \mspan{1}{3})$ and $(\mspan{4}{6}, \mspan{10}{12})$
  are also disjoint.
\end{example}

We can now define the UMDSDP condition, which will be sufficient to show
hardness:

\begin{definition}
    A single-variable domination rule $\domrule{}$ has
    \emph{unboundedly many disjoint strict domination pairs} (UMDSDP)
    if,
    given $n\in \mathbb{N}$, 
    we can compute in time polynomial in~$n$ 
    a document $d \in
    \alphabetstar$ and $n$ strict domination pairs $S_1, \ldots, S_n$ of
    $\domrule{}(d)$ that are pairwise disjoint.
\end{definition}

\begin{example}
  \label{exa:umdsdp}
	$\domrule{self}$ does not satisfy UMDSDP as it has no strict domination pairs.

The span length domination rule satisfies UMDSDP. Indeed,
		for $n \in \mathbb{N}$,
		we can take the word $a^n$
		and the disjoint strict domination pairs 
		$\{(\mspan{i}{i}, \mspan{i}{i+1}) \mid i \in [0, n-1]\}$.
		The same pairs show that UMDSDP holds for the span inclusion
                rule and for the left-to-right rule.
		
		Finally, the variable inclusion domination rule satisfies UMDSDP with the 
		set of pairs $\{(-, \mspan{i}{i+1}) \mid i \in [0, n-1]\}$ (remember the convention
		that \textquotedblleft$-$\textquotedblright~denotes that a variable is not assigned).

	Consider the domination relation $\domrule{start}$ defining the
        domination relation $\domrela{start}$ that contains the
        pairs $\{(\mspan{1}{i}, \mspan{1}{j})\mid i,j\in \mathbb{N}, i\leq j\}$
        plus the trivial pair $(-, -)$ for reflexivity.
        Then $\domrela{start}$ has unboundedly many strict domination pairs, but
        no two of them are disjoint, so the UMDSDP condition is not respected.
        (However, we will still be able to show intractability for this rule; see
        Proposition~\ref{prp:start}.)
\end{example}

We remark that, for single-variable domination rules that are regular, 
the UMDSDP condition holds whenever there \emph{exist}
arbitrarily many pairwise disjoint
strict domination pairs (i.e., in this case we can always efficiently compute them); see Appendix~\ref{sec:decideUMDSDP} for
details.

\begin{toappendix}

\subsection{UMDSDP for Domination Rules Expressed as Sequential VAs}\label{sec:decideUMDSDP}

In this section, we show that, for domination rules expressed as
  sequential VAs,
  the UMDSDP
  condition does not need to assert that the pairs can efficiently be computed,
  simply that they exist. Let us define a weaker condition than UMDSDP
  where we do not impose the efficiency requirement:

\begin{definition}
  We say that a single-variable domination rule $\domrule{}$ is
  \emph{weakly-UMDSDP} if, for any $n \in \NN$, there exists a document $d$ and
  $n$ strict domination pairs $S_1, \ldots, S_n$ of $\domrule{}(d)$ that are
  pairwise disjoint.
\end{definition}

  Thus, a single-variable domination rule $\domrule{}$ that is UMDSDP is in
  particular weakly-UMDSDP. Conversely, if $\domrule{}$ is weakly-UMDSDP and we
  can efficiently compute the document and the strict domination pairs,
  then it is UMDSDP.

  Our precise claim is that, for domination rules expressed as a VA, then being 
  weakly-UMDSDP implies being UMDSDP:

\begin{proposition}
  \label{prop:umdsdpconstr}
  Let $\domrule{}$ be a single-variable domination rule expressed as a sequential
  VA $\calA$. If $\domrule{}$ is weakly-UMDSDP, then it is UMDSDP.
\end{proposition}

  We need for this an additional claim on finite automata (without marker
  transitions), which is somewhat esoteric and somewhat
  difficult to show. We write $\Sigma$ the alphabet, and we say that a
  deterministic finite automaton on~$\Sigma$ is \emph{complete} if its
  transition function is a total function, i.e., for every state and letter $a
  \in \Sigma$ there is only one outgoing transition. This can be imposed in
  linear time without loss of generality by adding a sink state. We then have:

\begin{lemma}
  \label{lem:ramsey}
  For any deterministic complete finite automaton $\calA$ with set of states~$Q$ of
  size~$k$, there is a number $K$ depending only on~$k$ such that the following is true:
  Let $\delta\colon Q \times \Sigma \to Q$ be the (total) transition function
  of~$\calA$, and extend it to the function $\delta^*\colon Q \times \Sigma^*
  \to Q$ on words in the expected way.
  Let $w_1, \ldots, w_{K+1}$ be arbitrary words, and let 
  $q_1, \ldots, q_K \in Q$ be arbitrary states.
  Then there exist indices $1 \leq s <  t \leq K+2$ with $s<t$ 
  such that, letting
  $q \colonequals \delta^*(q_s, w_{s+1} \cdots w_{t-1})$, we have
  $\delta^*(q, w_s \cdots w_{t-1}) = q$.
\end{lemma}

  To prove Lemma~\ref{lem:ramsey}, we will use the following variant of Ramsey's theorem, see e.g.~\cite{graham1991ramsey} for a proof.
\begin{theorem}\label{thm:ramsey}
	For every pair $k,r\in \mathbb{N}$, there is an integer $R(k,r)$, such that for every $r$-coloring of the edges of a complete graph with at least $R(k,r)$ vertices, there is a monochromatic clique of size $k$, i.e., a clique whose edges all have the same color.~
\end{theorem}

\begin{proof}[Proof of Lemma~\ref{lem:ramsey}]
  Let $C$ be the set of functions from $Q$ to $Q$. Observe that $C$ is finite
  and $|C| = k^k$ depends only on $k$.
  To every word $w \in \Sigma^*$,
  we associate a function $f_w$ from $C$ defined in the following way: for
  each state $q \in Q$, we set $f_w(q) \colonequals \delta^*(q, w)$, i.e., the state that we reach when we start
  the automaton at~$q$ and read~$w$\footnote{We remark that this is the \emph{transition monoid}
  from algebraic automata theory.}.
  Let us choose $K> R(k+2, |C|)$ according to Theorem~\ref{thm:ramsey}.

  Consider now input words $w_1, \ldots, w_{K+1}$ as in the statement 
  of the lemma. We define the
  complete graph on $K+1$ vertices, each edge $\{i, j\}$ corresponding to the
  element $f_{w_i \cdots w_{j-1}}$ of~$C$ for the word $w_i \cdots w_{j-1}$.
  By the choice of $K$ and the statement of Theorem~\ref{thm:ramsey}, this graph has a 
  monochromatic $(k+2)$-clique, so there is a color $f \in C$ and positions
  $i_1 < \cdots < i_{k+2}$ such that for each $1 \leq a < b \leq k+2$,
  letting $w_{a,b} = w_{i_a} \cdots w_{i_b-1}$, then we have $f_{w_{a,b}} = f$.
  To prove the statement of the claim, we take $s
  \colonequals i_1$ and $t \colonequals i_{k+2}$ and define accordingly $q
  \colonequals \delta^*(q_s, w_{s+1} \cdots w_{t-1})$.

  Let us show the equality in the statement of the lemma.
  Consider the run~$\rho$
  where the automaton starts at state $q_s$, reads the word $w_{s+1} \cdots
  w_{t-1}$, and reaches state~$q$ by definition.
  Consider the states reached in this run just before reading the $k+1$ factors
  $w_{i_2}, \ldots, w_{i_{k+2}}$.
  By the pigeonhole principle, the same
  state~$q'$ is repeated twice, meaning there is some state~$q'$ such that we go
  from state~$q'$ to state~$q'$ when reading $w_{a,b} = w_{i_a} \cdots
  w_{i_b-1}$, i.e., $\delta^*(q', w_{a,b}) = q'$ for some
  $2 \leq a < b \leq k+2$. Now, our use of Theorem~\ref{thm:ramsey} ensures that,
  for any $1 \leq a' < b' \leq k+2$, we have $\delta^*(q', w_{a',b'}) = q'$.
  In particular, by considering the factor
  $w_{b, k+2} = w_{i_b} \cdots w_{i_{k+2}-1}$, we have the following (*): $\delta^*(q', w_{b,k+2})
  = q'$.

  Now,
  the run~$\rho$ starting at $q_s$
  and reading $w_{s+1}\cdots w_{t-1}$ does the following: first it reads $w_{s+1}
  \cdots w_{i_b-1}$
  and ends up in state~$q'$, then it reads $w_{b,k+2}$ and ends up in
  state~$q'$ by (*) above. 
  But by definition,
  on $w_{s+1} \cdots w_{t-1}$ starting in $q_s$, the automaton reaches the state~$q$, so
  it follows that $q = q'$. Considering now the word $w_s \cdots w_{t-1} =
  w_{1,k+2}$, 
  we know again that $\delta^*(q', w_{1,k+2}) = q'$, and thus in fact
  $\delta^*(q, w_{1,k+2}) = q$, 
   which is what we needed to show.
\end{proof}

We now show the main result of this section.

  \begin{proof}[Proof of Proposition~\ref{prop:umdsdpconstr}]

    We assume without loss of generality that~$\calA$ is \emph{input-output
    deterministic},
    in the sense that for every mapping in the output there is
    exactly one accepting run that witnesses it and, in that run in every state, there is only one possible continuation. It is shown in Theorem~3.1 of~\cite{DBLP:conf/pods/FlorenzanoRUVV18} how to construct VAs of this type from a general VA\footnote{We remark that strictly speaking~\cite{DBLP:conf/pods/FlorenzanoRUVV18} shows a transformation into so-called
    \emph{deterministic extended VAs}, but those can trivially be translated into VAs in our sense
    while preserving input-output determinism.}.
    We remark that this transformation might increase the size of the automaton exponentially
    and thus in particular takes time exponential in $|\calA|$. However, 
    this is not a problem: the claim that the domination rule is
    UMDSDP is about tractably computing a witnessing document and strict
    domination pairs from an input number~$n$, and the size of the VA defining
    the domination rule is a constant. Similarly, it will not be a problem that some of
    the arguments below are not constructive: we only show that $d$ exists and that there is
    a polynomial time algorithm to compute it but do not actually show how the algorithm
    can be inferred from $\calA$.

    Let $Q$ be the set of states of~$\calA$ and let $k$ be the size of $Q$. Let $K$ be the
    number given by Lemma~\ref{lem:ramsey} (the lemma applies to automata not
    VAs, but we will explain later how it is used in our context). As the rule
    defined by~$\calA$ is weakly-UMDSDP, there
    is a document $d$ featuring
    $(k+1) \times (K+2)+3$ strict domination pairs that
    are pairwise disjoint:
    we call them $S_1,
    \ldots, S_{(k+1)\times (K+2)+3}$ in increasing order of the left endpoints of the minimal size spans
    which cover them; note that these spans are pairwise disjoint.

    Let us consider the indices $i_1, \ldots, i_{(k+1)\times (K+2)}$ of~$d$
    which are the right endpoints of the spans covering the pairs $S_1, \ldots,
    S_{(k+1)\times(K+2)}$, i.e., for each $1 \leq j \leq (k+1)\times (K+2)$,
    the index $i_j$ is the position of the letter in~$d$ which is read by the
    first letter transition after we are done reading the last marker of
    pair~$S_j$. Let us then define
    $q_j$ 
    to be the state reached just before reading the letter at position $i_j$ in
    the witnessing run $\rho$ for the last pair
    $S_{(K+2)\times(k+1)+3}$. 
    As $\calA$ is input-output deterministic, we know that the run does not read
    any variable markers until it reaches the state
    $q_{(k+1)\times (K+2)}$ just before reading the letter at position
    $i_{(k+1) \times (K+2)}$;
    indeed all markers of~$\rho$ are
    assigned while reading the strict domination pair $S_{(K+2)\times(k+1)+3}$,
    and we know that the covering span of this strict domination pair starts
    after the right endpoint $i_{(k+1) \times (K+2)}$ of the domination pair
    $S_{(k+1)\times(K+2)}$. Here, we use the fact that of the three domination
    pairs $S_{(K+2)\times(k+1)+1}$, $S_{(K+2)\times(k+1)+2}$,
    $S_{(K+2)\times(k+1)+3}$, there are at most two strict pairs which assign an
    empty span, so the last pair must indeed be starting at a position strictly after
    $i_{(k+1) \times (K+2)}$. 
    Thus, thanks to input-output-determinism, $q_j$ is in fact the unique state
    that $\calA$ can reach by reading the document~$d$ until position $i_j$
    (excluded) and not assigning any markers.

    By the pigeonhole principle,
    there is a state $q$ and $K+2$ indices $j_1, \ldots, j_{K+2}$ such that for each $1 \leq a
    \leq K+2$ the state $q_{j_a}$ is $q$.

    Now, decompose $d = d_1 \ldots d_{K+3}$ accordingly,
    i.e., $d_1$ is the prefix of~$d$ until position~$j_1$ (excluded),
    for each $2 \leq a \leq K+2$ then $d_a$ is the factor of~$d$ from position
    $j_{a-1}$ (included) to $j_{a}$ (excluded), and $d_{K+3}$ is the suffix of~$d$
    starting in position $j_{K+2}$ (included). Our choice of factors ensures that, when
    it does not assign any markers, then $\calA$ goes from the initial state
    to~$q$ when reading $d_1$, and (*) goes from~$q$ to~$q$ when reading~$d_a$ for any
    $2 \leq a \leq K+2$. (It can also finish the run from~$q$ by
    reading~$d_{K+3}$, assigning markers for the pair $S_{(K+2)\times(k+1)+3}$;
    but we will not use this.)

    Let us now consider, for each $2 \leq a \leq K+2$, the run $\rho_a$ that
    assigns the domination pair $S_{j_a-1}$. The run $\rho_a$ (1.) goes from the
    initial state to state~$q$ while reading $d_1 \cdots d_{a-1}$, then (2.) it reads
    $d_a$ whose last pair is $S_{j_a-1}$, and assigns markers while reading that
    pair. Let $q'_a$ be the state reached just before the following letter
    transition, i.e., the letter transition that reads the letter at position
    $i_{j_a}$. We know (3.) that the automaton can finish the run, i.e., $\rho_a$
    continues reading from state $q'_a$ the document $d_{a+1} \cdots d_{K+3}$
    without assigning any markers, and reaches an accepting state.

    Hence, let us define $\calA'$ the deterministic complete finite automaton
    obtained from~$\calA$ simply by dropping all marker transitions, and let us
    use Lemma~\ref{lem:ramsey} on~$\calA'$ (up to offsetting all indices by~$1$). We know that, for 
    a choice of $2 \leq s <
    t \leq K+2$, 
    considering (**) the state $q''$ (called~$q$ in the statement of
    Lemma~\ref{lem:ramsey}) reached by reading $d_{s+1} \cdots d_{t-1}$ from
    the state $q'_s$ (intuitively, just after assigning the markers), then  (***)
    reading $d_s \cdots d_{t-1}$ in~$\calA'$ goes from~$q''$ to~$q''$. Thus, coming
    back to~$\calA$, we know the following:
    \begin{enumerate}
      \item The run $\rho_s$ goes from the initial state
        to state~$q$ while reading $d_1 \cdots d_{s-1}$ (from~(1.) above),
      \item The run~$\rho_s$ then assigns the markers of
    the pair $S_{j_s-1}$ while reading $d_s \cdots d_{t-1}$ and goes to
        state~$q''$ (from~(2.) above on~$d_s$ and then (**) above on $d_{s+1}
        \cdots d_{t-1}$),%
  \item The run~$\rho_s$ then
    finishes without assigning any variables (by (3.) above).
  \item The VA $\calA$, when starting at state~$q''$ and reading $d_s \cdots
    d_{t-1}$ without assigning any markers, goes to
        state~$q''$ (by~(***) above);
  \item The VA $\calA$, when starting at state~$q$ and reading $d_s \cdots
    d_{t-1}$ without assigning any markers, goes to state~$q$ (by~(*)
        above).
    \end{enumerate}

    We are now finally ready to define a decomposition of~$d$ as $d = u v w$ where $v$
    consists of that subword $d_s \cdots d_{t-1}$ and $u$ and $w$ are respectively the prefix and
    suffix that precede and follow it.

    We claim that, for any $n \in \NN$, the VA $\calA$ on the document
    $uv^nw$ has $n$ strict domination pairs that are disjoint and can
    be easily computed. More specifically, there
    is one of the domination pairs in the factor $v^n$ in every copy of $v$ at the positions
    corresponding to the endpoints of the pair $S_{j_s-1}$.
    If this is true, then
    it establishes that the domination rule of $\calA$ is UMDSDP: given an
    integer $\NN$, we can compute the document $u v^n w$ 
    in polynomial time in the value of~$n$. Moreover, by some arithmetic on the 
    lengths of $u$ and $v$ and the position of the considered pair in $v$, we
    can easily compute the $n$ pairs in polynomial time.

    To see why we have the $n$ pairs, pick $n \in \NN$ and choose one index $1 \leq i
    \leq n$ and show why we can obtain the domination pair. By construction,
    when $\calA$ reads the prefix~$u$ without assigning any marker, then it goes
    to state~$q$ (point 1 above). Further, when it reads~$v$ without assigning any markers, it
    also goes to~$q$ (point 5 above).
    Hence, the same is true when reading $v^{\ell-1}$ for any $\ell\in [n]$. Now, we know that we can
    read~$v$ and assign markers corresponding to a strict domination pair within
    that factor, and go to state~$q''$ (point 2 above). Further,
    reading $v$ without assigning any markers goes from state~$q''$ to
    state~$q''$ (point 4 above), allowing us to continue the run without taking any marker transitions until we have read the last copy of $v$; and finally we
    can read~$w$ without
    assigning any markers and complete the run  to an accepting run (point 3 above). 
    This concludes the proof.
  \end{proof}
\end{toappendix}

\subparagraph*{UMDSDP implies hardness.}
We now show that the UMDSDP condition implies that skyline extraction is hard.
The proof is a variant of the one for variable
inclusion:

\begin{toappendix}
	\subsection{Hardness Proofs with the UMDSDP Condition}
\end{toappendix}

\begin{theoremrep}
\label{thm:maindicho}
    Let $\domrule{}$ be a single-variable 
    domination rule satisfying UMDSDP.
    The skyline extraction problem for~$\domrule{}$, given a 
    sequential VA $\calA$ and
    a document $d \in \alphabetstar$,
    is not output-polynomial unless $\mathsf{P} = \mathsf{NP}$.
\end{theoremrep}

\begin{proof}
	As in the proof of Theorem~\ref{thm:nphardspanner}, we again reduce from
        the satisfiability problem for CNF-formulas.

    Let $\domrule{}$ be a variable-wise domination rule satisfying UMDSDP.
    Let $F$ be a CNF formula, with $n_v$ Boolean variables $x_i$ and $n_c$ clauses $C_j$.
    We will show how to reduce satisfiability of $F$ to computing the skyline
    with respect to the rule $\domrule{}$ of a suitably constructed regex formula.
    To this end, we will construct a spanner on the variables $v_{i,j}$ for $i\in [n_v]$ and $j\in [n_c]$.
    We define the sets $T_i \colonequals \{j \mid $ $x_i$ appears positively in $C_j\}$,
    and $F_i \colonequals \{j \mid $ $x_i$ appears negatively in $C_j\}$.

    The idea of the proof is similar to that of Theorem~\ref{thm:nphardspanner} which we first sketch here again:
    We define two spanners $r_\valid$ and $r_\mask$.
    The spanner $r_\valid$ will again capture 
    one mapping $m$ for each possible CNF variable assignment $a$.  
    Each mapping $m$ will encode which clauses $C_j$ get satisfied by which variable $x_i$ 
    under the assignment $a$.
    The second spanner $r_\mask$ will again capture $n_c$ other mappings $m_j$
    each one dominating all mappings $m$ captured by $r_\valid$ encoding an assignment 
    which does not satisfy clause $C_j$.
    Together the mappings $m_j$ dominate all mappings captured by 
    $r_\valid$ encoding non-satisfying assignments.
    Then the CNF formula $F$ is satisfiable if and only if 
    there is a mapping $m$ captured by $r_\valid$ in the skyline of $r_\valid
    \lor r_\mask$.
    
	We now give the details of the construction.
    We first compute the input document and spans we will use.
    Since $\domrule{}$ satisfies UMDSDP,
    we can compute in polynomial time in~$n_v$ a document $d \in \alphabetstar$ 
    and 
    $n_v$ pairwise disjoint strict domination pairs $(s_1^1, s^1_2), (s_1^2, s^2_2), \ldots, (s^{n_v}_{1}, s_2^{n_v})$ in $\domrule{}(d)$.
    We cut the document $d$ into $n_v$ disjoint parts $d_1, \ldots, d_{n_v}$, such that
    for each $i\in [n_v]$ the part $d_i$ contains the strict domination pair $(s_1^i,s_2^i)$. This is possible, because by definition the $(s_1^i,s_2^i)$ do not overlap.
    Let $s^i$ be the span which gives the limits of $d_i$, i.e., $d_{s^i}= d_i$.
    Note that the spans $s_1^i$ and $s_2^i$, if defined, are included in the span $s^i$ for all $i\in [n_v]$.

    We now construct the spanners, starting with $r_\valid$.
    Given a variable index $i \in [n_v]$, 
    a set of clause indices $I \subseteq [n_c]$,
    and an index~$b \in \{1, 2\}$,
    we first define the spanner $r^i(I,b)$. The spanner will read the
    word~$d_i$, and will assign the variables $v_{i,j}$ for $j \in I$ to
    according to the value~$s_b^i$ in the domination pair. Formally, if $s_b^i$
    is not defined, then we read $d_i$ without assigning anything:
    \begin{align*}
      r^i(I, b) \colonequals d_i
    \end{align*} 
    and if $s_b^i$ is defined then writing $s_b^i = [p,q\rangle$, writing $[p',q'\rangle$ for the
    span that defines $d_i$, we have:
    \begin{align*}
      r^i(I, b) \colonequals d_{\mspan{p'}{p}} \left({\Join_{j \in I}} v_{i,j}\{
        d_{\mspan{p}{q}} \} \right)
      d_{\mspan{q}{q'}})
    \end{align*} 
    Note that the join here is easy to express for each concrete $s^i_b$ since all variables $v_{i,j}$ with $j\in I$ are mapped to the same span.

    Now, to construct $r_\valid$,
    we define the spanners $r_\valid^i$ for $i\in [n_v]$ 
    which, when reading $d_i$, captures two mappings $m_t, m_f$ each.
    The mapping $m_t$ (resp., $m_f$) encodes the clauses $C_j$
    satisfied when $x_i$ is assigned to true (resp., false) as follows:
    \begin{align*}
      r_{\valid}^i \colonequals ( r^i(T_i, 2) \Join r^i([n_c] \setminus T_i, 1))
    \lor ( r^i(F_i, 2) \Join r^i([n_c] \setminus F_i, 1)).
    \end{align*}
	Again, the construction is easy for each concrete situation, since all variables concerned are only mapped to two different spans $s^i_1$ and $s^i_2$, so $r^i_\valid$ can be easily expressed as a VA and also be computed in polynomial time.
    We take the concatenation of the spanners $r_\valid^i$ and obtain
    $r_\valid$ which reads the entire document $d$ and encodes 
    all possible assignments of the Boolean variables, so
    \begin{align*}
      r_\valid \colonequals r_\valid^1 \cdot \ldots \cdot  r_\valid^{[n_v]}.
    \end{align*}
    Since all the $r^i_\valid$ can be constructed in polynomial time and concatenation is easy to express for VAs, $r_\valid$ is a regular spanner whose VA can be computed in polynomial time.

    To construct a CNF variable assignment 
from a mapping $m$ captured by $r_\valid$ 
we will use the following encoding:
If the mapping $m$ assigns the variable~$v_{i,j}$ following $s_2^i$ (i.e., to
this span if it is a span, or not at all if it is undefined), then
the variable $x_i$ appears in the clause $C_j$ 
and $x_i$ is assigned so as to make that clause $C_j$ true.
Thanks to the way in which the $r_\valid^i$ are constructed,
the mapping gives a consistent choice of how to assign all the occurrence of any
given variable.

    Now let us construct $r_\mask$.
    We define the spanners $r_\mask^j$ for $j \in [n_c]$, 
    which read $d$ and capture one mapping $m^j$ each.
    The mapping $m^j$ will be assigned following the $s_2^i$ almost everywhere
    except for the variables $v_{i,k}$ where $k = j$ where it will be assigned
    following $s_1^i$.
    Recall that $s_2^i$ dominates $s_1^i$.
    \begin{align*}
      r_\mask^j \colonequals {\cdot_{i \in [n_v]}} (r^1(j, s_1^1) \Join r^1([n_c] \setminus j, s_2^1) ) \cdot \ldots \cdot (r^{[n_v]}(j, s_1^{[n_v]}) \Join r^{[n_v]}([n_c] \setminus j, s_2^{[n_v]}) )
    \end{align*}
     	where, as before, $\cdot$ denotes concatenations. As above, the $r^i(j, s_1^i)$ and $r^i([n_c], s_2^i)$ can be turned into a VA in polynomial time, so a VA for $r^j_\mask$ can be constructed in polynomial time as well.
    We take the union of the spanner $r_\mask^j$ and obtain 
    \begin{align*}
    r_\mask \colonequals \bigcup_{j\in [n_c]} r_\mask^j
	\end{align*}
which reads the entire document $d$. 
$r_\mask$ captures the $n_c$ mappings which together dominate all mappings 
captured by $r_\valid$ encoding a non-satisfying assignment.    Again, this can be efficiently encoded as a VA, and so $r_\valid \lor r_\mask$ can also be encoded in a VA in polynomial time.

    Let us now prove that the reduction is correct, 
    that is that $F$ is satisfiable if and only if 
    the skyline of $r_\valid \lor r_\mask$ has a mapping $m$ captured by
    $r_\valid$.
    Assume first that $F$ is satisfiable.
    Let $v$ be a satisfying assignment to the $x_i$.
    Then there is a corresponding mapping captured by $r_\valid$ encoding this satisfying assignment. 
    As $v$ satisfies all clauses,
    for every clause index $j \in [n_c]$ 
    there is a CNF variable index $i \in [n_v]$ 
    such that $v_{i,j}$ is assigned following $s_2^i$.
    Hence $m$ is not dominated by any mapping captured by $r_\mask$.
    This means that the skyline contains a mapping captured by~$r_\valid$: either $m$ or some other mapping which dominates $m$ and which by transitivity must also be a mapping captured by~$r_\valid$.

    Now assume a mapping $m$ captured by $r_\valid$ 
    is in the skyline of $r_\valid \lor r_\mask$.
    As $m$ is not dominated by any mapping captured by $r_\mask$,
    for all clause indices $j \in [n_c]$ 
    there must exists a variable index $i \in [n_v]$ 
    such that $v_{i,j}$ is assigned following $s_2^i$.
    Therefore the CNF assignment encoded by $m$ satisfies all clauses
    and thus it witnesses that $F$ is satisfiable.
    Overall, we have that $F$ is satisfiable if and only if 
    a mapping $m$ captured by $r_\valid$ is in the skyline of
    $r_\valid \lor r_\mask$.

    To check if a mapping $m$ in the skyline of $r_\valid \lor r_\mask$
    is captured by $r_\valid$ it is sufficient to check 
    that all clauses indices $j\in [n_c]$ have some variable index $i\in [n_v]$ such that 
    $v_{i,j}$ is assigned following $s_2^i$, which can be checked in linear time. 
     
     We complete the proof by showing that, if we can solve in output-polynomial time the skyline extraction problem for $\domrule{}$ on the sequential VA and document defined in the reduction (as is assumed in the theorem statement), then we have $\mathsf{P}= \mathsf{NP}$, concluding the proof. 
     This is a variant of Lemma~\ref{lem:outputpoly}, except that, now, we do not know the exact number of captured mappings (as, e.g., some mappings captured by $r_\mask$ may dominate each other).

     Assume that such an algorithm exists and it solves the problem in time $p(|P|+|d|+|\mathit{out}|)$ where $\mathit{out}$ is the output. We solve SAT in polynomial time as follows: given a CNF $F$, construct $P\colonequals r_{\valid}\lor r_{\mask}$
     and $d$ as above. Run the algorithm for skyline computation for $p(|P|+|d|+ n_c |m|)$ steps where $|m|$ is the encoding size of a single mapping.
     If this terminates in the allotted time, check if the output contains a mapping $m$ encoding a satisfying assignment and answer the satisfiability of $F$ accordingly. If the algorithm does not terminate, answer that $F$ is satisfiable. This is correct, because, similarly to the reasoning for Lemma~\ref{lem:outputpoly}, in that case the skyline contains at least $n_c+1$ mappings out of which only $n_c$ can be captured by $r_{\mask}$. So there is at least one mapping $m$ captured by $r_{\valid}$ in the skyline and it follows that $F$ is satisfiable.
\end{proof}

This implies the hardness of the other variable-wise domination rules
presented earlier, completing Corollary~\ref{cor:varinc}. Note that these
rules are schema-based spanners, and we can also notice that hardness
already holds if the input spanner is functional, i.e., schema-based:

\begin{corollaryrep}
  \label{cor:schemabased}
  There is no algorithm for the skyline extraction problem with respect to the
  span inclusion domination rule or the left-to-right domination rule or the span length domination rule which is
  output-polynomial in combined complexity, unless $\mathsf{P} = \mathsf{NP}$.
  This holds even if the input VA is required to be functional.
\end{corollaryrep}

\begin{proof}
  We have shown in Example~\ref{exa:umdsdp} that span inclusion and span length domination rules satisfy UMDSDP, so the result immediately follows from Theorem~\ref{thm:maindicho}.

  For the second part of the claim, if the input VA is required to be functional, then we can conclude by inspection of the proof of the theorem. Specifically, we know that the strict domination pairs for the span inclusion and span length domination relations always assign both the left-hand-side and right-hand-side span. Now, the spanners defined in the proof always assign all variables of the form~$v_{i,j}$ in the mappings that they capture: this can be seen on the $r_\valid^i$ for each value of~$i$ and on the $r_\valid^j$ for each~$j$. Thus, the constructed VAs are functional, and so hardness holds even in that case.
\end{proof}

\begin{toappendix}
  \subsection{Proofs for Variable Inclusion-Like Rules}
\end{toappendix}
\subparagraph*{Variable Inclusion-Like Rules.}
We have seen that the UMDSDP condition is a sufficient condition for skyline
extraction to be hard, but this leaves open the question of whether it is
necessary. We will now focus on a fragment of domination rules which we call \emph{variable inclusion-like domination rules}, where this is the
case. Formally, we
say that a domination relation $\domrela{}$ is \emph{variable inclusion-like} if
for all strict domination pairs $(m_1, m_2)$ we have for all $x\in \Vars$ that if $m_1(x)$
is defined, then $m_2(x)$ is defined as well and $m_1(x)= m_2(x)$.

In contrast with the variable inclusion rule
that
contains all such pairs $(m_1, m_2)$, we only require that a subset of them hold
in $\domrela{}$.
We will define variable inclusion-like domination rules in a variable-wise
fashion: for single-variable variable inclusion-like rules, the strict
domination pairs are necessarily of the form $(-, s)$ for a span~$s$. In other
words, a variable-wise inclusion-like domination rule is 
defined by indicating, on each document, which spans~$s$ can appear as the
right-hand-side of such a pair. Further,
for variable inclusion-like rules, two strict domination pairs are
disjoint if and only if their right-hand-sides are.

We can show that, on variable inclusion-like domination rules,
we have a dichotomy on a variant of the UMDSDP condition:
\begin{theoremrep}\label{thm:dichovarinclike}
  Let $\domrule{}$ be a single-variable domination rule which is variable
  inclusion-like. If $\domrule{}$ satisfies the UMDSDP condition or accepts a
  pair of the form $(-, \mspan{i}{i})$ on some document, then the
  skyline extraction problem for~$\domrule{}$, given a sequential VA and
  document, is not output-polynomial in combined complexity unless $\mathsf{P} = \mathsf{NP}$.
  Otherwise, the skyline extraction problem for~$\domrule{}$ is
  output-polynomial in combined complexity.
\end{theoremrep}

The lower bound of the dichotomy follows from
Theorem~\ref{thm:maindicho}, plus the observation that a single pair of the form $(-,
\mspan{i}{i})$ is sufficient to show hardness:

\begin{toappendix}
  We prove this result in the rest of the appendix. We first show the lemma to exclude the case where we include a pair featuring an empty span:
\end{toappendix}

\begin{lemmarep}\label{lem:sizezero}
	Let $\domrule{}$ be a single-variable domination rule that accepts on some document a pair $(-, \mspan{i}{i})$.
        Then the skyline extraction problem for~$\domrule{}$, given a sequential
        VA and document, is not 
        output-polynomial in combined complexity unless $\mathsf{P} = \mathsf{NP}$.
\end{lemmarep}

\begin{proof}
  The proof of this result is an easy variation of Theorem~\ref{thm:nphardspanner}. Let $d$ be a document and $(-, \mspan{i}{i})$ be the accepted pair: let $d_1$ and $d_2$ be respectively the prefix and suffix of~$d$ before position~$i$, so that $d = d_1 d_2$.
  We do the same reduction as in Theorem~\ref{thm:nphardspanner}, but running it on the document $d$, and defining $r_\valid$ and $r_\mask$ by adding the prefix~$d_1$ and the suffix~$d_2$. Given that the reduction can only assign the span $\mspan{i}{i}$ to variables, the rest of the argument is unchanged.
\end{proof}

Hence, the interesting result in Theorem~\ref{thm:dichovarinclike} is the upper
bound.
We show it in the appendix by observing that the set of right-hand-sides of
strict domination pairs for variable inclusion-like rules that do not satisfy
UMDSDP have bounded hitting
set number, and showing that this implies tractability.

\begin{toappendix}

  Hence, in the rest of this section, we show Theorem~\ref{thm:dichovarinclike} in the case where no such pair is accepted.

  We start with some easy observations on sequential VAs. 
Let $\mathcal A$ be a VA in variables $X$. Let $d$ be an input to $\mathcal A$. At any moment in a run of $\mathcal A$ on $d$, we say that $x\in X$ is \emph{open} if $\vop{x}$ has been read but $\vcl{x}$ has not been read.

We remark that a similar concept is used under the name \emph{variable configurations} in~\cite{DBLP:conf/pods/FreydenbergerKP18} in the context of enumeration algorithms for so-called functional VAs. We here only need the following easy observation:

\begin{observation}\label{obs:openvariables}
	Let $s$ be a state of $\mathcal A$.
        If there is an accepting run of $\mathcal A$ on an input $d$ such that a variable $x$ is open at a moment in the run in which the automaton $\mathcal A$ is in state $s$, then for any accepting run over any document, at any point where the automaton reaches state~$s$, then the variable $x$ is open.
\end{observation}
\begin{proof}
	By way of contradiction, assume this were not the case. Let $\rho_1$ be an accepting run on the input~$d$ on which $x$ is open at some point when reaching the state $s$: let $\rho_1'$ be the suffix of~$\rho_1$ after that point. Let $\rho_2$ be an accepting run on a document $d_2$ in which $x$ is not open when reaching the state $s$ at some point: let $\rho_2'$ be the prefix of~$\rho_2$ until that point.
        Let us built $\rho' = \rho_2' \rho_1'$: this gives an accepting run of~$\calA$ on some document. Now, as $\calA$ is sequential, we know that $\rho'$  will contain the closing marker for~$x$ and no opening marker for~$x$. As $\rho_2'$ either contained an opening and closing marker or contained no marker, then $\rho'$ is not valid, contradicting the assumption that $\calA$ is sequential. This concludes the proof.
\end{proof}

We make some additional definitions. Let $S$ be a set of spans that does not
  contain any span of length~$0$. We call a set $H\subseteq \mathbb{N}$ a \emph{hitting set} of $S$ if for every $s\in S$ we have that there is an $i\in H$ such that $\mspan{i}{i+1}$ is contained in $s$. We say that the \emph{hitting set number} of $S$ is the smallest number of a hitting set of $S$.

  Let us now consider a document~$d$, and consider the set of mappings $D(d)$ of the single-variable variable
  inclusion-like domination rule $D$ on the document~$d$, and the domination relation
  $\domrela{}$ that it defines.
%
%
%
  Consider the set of strictly dominating
  spans defined in the following way:
\begin{align*}
  \pi_2(\domrela{}, d) \colonequals \{m_2 \mid (m_1, m_2)\in {\domrela{}}, m_1(x)\ne m_2(x)\}
\end{align*}
  In the above, necessarily $m_1$ is undefined on $x$ whenever $m_1(x)\ne m_2(x)$, because the relation is
  variable inclusion-like. In other words, $\pi_2(\domrela{}, d)$ is simply the
  set of spans that dominate the choice \textquotedblleft$-$\textquotedblright~of not assigning a variable.
  We know from an earlier assumption that $\pi_2(\domrela{}, d)$ never contains an empty span,
  as otherwise we could have concluded by Lemma~\ref{lem:sizezero}.
Now, we define the hitting set number of $\domrule{}$ on $d$ to be that of $\pi_2(\domrela{}, d)$. Finally, we define the hitting set number of $\domrule{}$ to be the supremum of the hitting set number on $d$ taken over all documents~$d$.

Hitting set numbers are interesting because they bound the size of extractions results of spanners compared to their skylines in the following way.

\begin{lemma}\label{lem:sizebound}
	Let $\domrule{}$ be a variable inclusion-like domination rule whose
        right-hand-sides never contain an empty span, and whose hitting set
        number is $k\in \mathbb{N}$.
        Let $P_{\mathcal A}$ be a spanner defined by a sequential VA $\mathcal A$. Then we have for every document~$d$
	\begin{align*}
          |P_{\mathcal A}(d)| \le |\mathcal A|^k |(\domoper{\domrule{}} P_{\mathcal A})(d)|.
	\end{align*}
\end{lemma}
\begin{proof}
  We will show that every mapping in $(\domoper{\domrule{}} P_{\mathcal A})(d)$ can only dominate $|\mathcal A|^k$ mappings in $P_{\mathcal A}(d)$.
	The lemma follows directly from this.

        We have that $\pi_2(\domrule, d)$ has a hitting set $H$ of size $k$, so
        let $K=\{i_1, \ldots, i_k\}$ be such a hitting set.

        Consider $m\in (\domoper{\domrule{}} P_{\mathcal A})(d)$ and $m' \in P_{\mathcal A}(d)$ such that $m'\domrela{}m$. Let $X$ be the variables on which $m$ is defined. Then we can construct subsets $X_1, \ldots, X_k$ of $X$ such that for every $j\in [k]$ we have that $x\in X_j$ if and only if $\mspan{i_j}{i_j+1}$ is contained in $m(x)$. Since $H$ is a hitting set, we get that $X =  \bigcup_{j\in [k]} X_j$.
        Since $m'\domrela{}m$ and because $\domrela{}$ is variable
        inclusion-like,
        the mapping $m'$ is defined on a subset of $X$ and for all variables $x$ on which it is defined, we have that $m'(x) = m(x)$, so the only difference between $m$ and $m'$ is on which variables they are defined.
	
	Now consider $j\in [k]$, and let $X_j^{m'}\colonequals \{x\in X_j\mid m' \text{ is
        defined on } x\}$.
	Consider the state $s$ that $\mathcal A$ is in when it reads a letter transition for the $i_j$-th letter of $d$ in an accepting run that 
        yields the mapping~$m'$.
        In the state $s$, the open variables must be exactly those in $X_j^{m'}$, since $\mspan{i_j}{i_j+1}$ is contained in those spans. By Observation~\ref{obs:openvariables}, since every set of open variables must have its own state, it follows that the number of subsets $X_j^{m'}$ that $m'$ can choose from is bounded by the number of states of $\mathcal A$, so $|\mathcal A|$. Reasoning the same for all $j\in [k]$, we get that $\mathcal A$ can only encode $|\mathcal A|^k$ different mappings $m'$ dominated by $m$, so the claim follows.
\end{proof}

\begin{corollary}\label{cor:outputpolynomial}
	Let $\domrule{}$ be a variable inclusion-like domination rule capturing no empty span as the right-hand-side of a strict domination pair, and assume that the hitting set number is bounded by a constant. Then, given a spanner $P_{\mathcal A}$
        defined as a sequential VA $\calA$ and a document $d$, the skyline
        $(\domoper{\domrule{}} P_{\mathcal A})(d)$ can be computed in output polynomial time.
\end{corollary}
\begin{proof}
  We simply compute $P_{\mathcal A}(d)$ which for sequential VAs can be done in output-polynomial time, see~\cite{constdelay,PeterfreundFKK19}. We then explicitly compare all pairs of mappings in $P_{\mathcal A}(d)$ to filter out those that are dominated and thus compute $(\domoper{\domrule{}} P_{\mathcal A})(d)$. Since, by Lemma~\ref{lem:sizebound}, the intermediate result $P_{\mathcal A}(d)$ is only polynomially larger than the result $(\domoper{\domrule{}} P_{\mathcal A})(d)$. Further, the domination relation can be checked in polynomial time: it suffices to determine for each variable whether the assignment done in two mappings are related by the single-variable domination rule, and this is simply testing acceptance of the ref-word (see Appendix~\ref{app:prelim}) into the VA that defines the rule (note that, because multiple orders are possible for the markers, we may need to test acceptance of constantly many ref-words).
	Thus, the overall process runs in polynomial time.
\end{proof}

We will now show that the upper bound of Corollary~\ref{cor:outputpolynomial} corresponds to the case where the UMDSDP property does not hold. To this end, we show the following correspondence.
\begin{lemma}\label{lem:hittingdisjoint}
	Let $S$ be a finite set of spans not containing a span of length $0$ and
        let $\domrela{S}\colonequals \{(-, s)\mid s\in S\}$ be the variable inclusion-like domination relation induced by $S$. Let $k_h$ be the hitting set number of $S$ and let $k_p$ be the maximum number of strict dominating pairs of $\domrela{S}$. Then  we have:
	\begin{align*}
		k_p \le k_h \le 2k_p.
	\end{align*}
\end{lemma}
\begin{proof}
	Note that two pairs in $\domrela{S}$ are disjoint if and only if their second coordinates are disjoint, so instead of disjoint pairs it suffices to argue on disjoint spans in $S$.
	
	For the first inequality, observe that whenever two spans contain the same span $\mspan{i}{i+1}$, by definition they are not disjoint. The claim follows directly: given a set of more than $k_h$ spans, each span contains one of the elements of the hitting set, so then two spans must contain the same element, so they are not disjoint, and thus the set is not disjoint.
	
        For the second inequality, let $K= \{i_1, \ldots, i_k\}$ be a hitting set of $S$ of minimal size. For each $j\in k$, let $S_j$ be the subset of the spans of $S$ that contain $\mspan{i_j}{i_j+1}$.
        We claim that for every $j$ we have a span in $S_j$ that is neither in $S_{j-1}$ nor in $S_{j+1}$ (we slightly abuse notation here and say that $S_{j'}$ is empty for indices $j'\notin [k]$).
        By way of contradiction, assume this were not the case, then $i_j$ could be deleted from $K$ resulting in a smaller hitting set of $S$, which contradicts the assumption that $K$ is minimal.
        This directly lets us choose a set $S'$ of spans that are pairwise disjoint of size $\lceil k_h/2 \rceil$:
        take a span which is in $S_1$ but not in $S_2$, take a span which is in $S_3$ but not in $S_2$ (and hence does not intersect the first span) or $S_4$, etc. It follows that $k_h \le 2 k_p$.
\end{proof}

We can now finally show Theorem~\ref{thm:dichovarinclike}.

\begin{proof}[Proof of Theorem~\ref{thm:dichovarinclike}]
  For the upper bound, we know that if the rule does not satisfy the UMDSDP, then it also is not weakly-UMDSDP (by contrapositive of Proposition~\ref{prop:umdsdpconstr}), and hence the maximal number of strict dominating pairs on documents are bounded. By Lemma~\ref{lem:hittingdisjoint}, and as the rule does not capture an empty span as a right-hand-side of a strict dominating pair, the maximal size of a hitting set on a document is also bounded by a constant, say~$k$. (As this number only depends on the fixed rule, it does not matter whether we can efficiently compute it.)

  Now, given a VA $\calA$ and document $d$, we can compute the result in output-polynomial time, 
  simply using Corollary~\ref{cor:outputpolynomial}. Note that for this we do not need to compute the hitting set of cardinality~$k$, as its goal is just to bound the factor between the number of results and the number of maximal results.

  For the hardness part, there are two cases: if there is a document $d$ such that $\domrule{}(d)$ contains a pair of the form $(\varepsilon, \mspan{i}{i+1})$, then the result follows from Lemma~\ref{lem:sizezero}. If such a document does not exist, then as the rule does not satisfy the UMDSDP, 
  we conclude by Theorem~\ref{thm:maindicho} that the evaluation is not output-polynomial.
\end{proof}

\end{toappendix}

\begin{toappendix}
  \subsection{Proofs for Other Cases}
\end{toappendix}

\subparagraph*{Other cases.}
Theorems~\ref{thm:maindicho} and~\ref{thm:dichovarinclike} do not settle the complexity of
non-UMDSDP domination rules which are
not variable inclusion-like. We conclude with
some examples of rules that can be shown to be intractable. We first show it
for the rule $\domrela{start}$ introduced earlier:

\begin{propositionrep}
  \label{prp:start}
  Refer back to the variable-wise domination rule $\domrule{start}$ from
  Example~\ref{exa:umdsdp}. There is no output-polynomial combined complexity
  algorithm for the skyline extraction problem for that rule, assuming
  $\mathsf{P} \neq \mathsf{NP}$.
\end{propositionrep}

\begin{proof}
  We follow a similar approach to the previous hardness proofs. We reduce from the satisfiability problem for CNFs. Given a CNF $\Phi$ with $n_x$ Boolean variables $x_i$ and $n_c$ clauses $C_j$, and let us
    construct a spanner on variables $v_{i,j}$ for $i\in [n_x]$ and $j\in [n_c]$. For each variable $i$ let $T_i$ and $F_i$ be the subset of clauses that will be satisfied if setting $x_i$ to true (resp., to false).

    The document used is $d = a^{2 n_x}$.
    The spanner $r_\valid$ is directly defined as a VA: first open all variables. Immediately close variables of the form $x_{i,j}$ where variable $i$ does not occur in clause~$j$.
    Then, for each variable $1 \leq i \leq n_x$, do the disjunction between closing all variables $x_{i,j}$ for all $j \in F_i$ and then reading~$aa$, and between reading $a$, closing all variables $x_{i,j}$ for all $j \in T_j$, and then reading~$a$. Unsurprisingly, the mappings captured by $r_\valid$ are in correspondence with valuations of the variables, where variable $i$ is set to false if it captures the span $[0, 2(i-1)\rangle$ and to true if it captures the span $[0,2(i-1)+1\rangle$.

    The spanner $r_\mask$ is a disjunction between $n_c$ spanners, each capturing a mapping: for $1 \leq k \leq n_c$, the $k$-th spanner captures the mapping where we assigns the variables $x_{i,j}$ for all $i\in[n_x]$ and $j \in [n_c]\setminus \{k\}$ to the entire document (span $[0, 2n_x\rangle$), and assign the variables $x_{i,k}$ for all $i \in [n_x]$ to the empty span $[0, 0\rangle$.

    As usual, a mapping of $r_\valid$ is strictly dominated by the $k$-th mapping of $r_\mask$ if it did not assign any variable of the form $x_{i,k}$, that is, if in the choice of valuation there was no variable making clause~$k$ true. Hence, a mapping of $r_\valid$ is not strictly dominated by any mapping of $r_\mask$ if it represents a satisfying assignment of~$\Phi$.

    We accordingly claim that if the skyline of $r_\valid \lor r_\mask$ contains a mapping captured by $r_\valid$, then $\Phi$ is satisfiable. Indeed, for the forward direction, a mapping of $r_\valid$ in the skyline cannot be strictly dominated by any mapping of $r_\mask$ so it witnesses that $\Phi$ is satisfiable. For the backward direction, we have argued that a satisfying assignment of the formula describes a mapping of $r_\valid$ which is not dominated by any mapping of~$r_\valid$, and this means that this mapping or some other mapping of~$r_\valid$ must be in the skyline.

    We then conclude with the argument at the end of the proof of Theorem~\ref{thm:maindicho} that the skyline extraction problem for that rule cannot be in output-polynomial combined complexity unless $\mathsf{P} = \mathsf{NP}$.
\end{proof}

We show hardness for another rule that fails the UMDSDP, where all strict domination
pairs share the same right-hand-side:

\begin{propositionrep}
  Consider the variable-wise domination rule expressed by the regular expression
  $x\{a^*\} a^* \dagg{x}\{b\} \lor \domrule{self}$. There is no output-polynomial combined complexity
  algorithm for the skyline extraction problem for that rule, assuming
  $\mathsf{P} \neq \mathsf{NP}$.
\end{propositionrep}

\begin{proof}
  The proof is again a variant of the preceding proofs. We reduce again from the satisfiability problem for a CNF $\Phi$ with $n_x$ variables $x_i$ and $n_c$ clauses $C_j$, again writing $T_i$ and $F_i$ the clauses satisfied when making $x_i$ true or false respectively.

  The document is $d = a^{2 n_x} b$. The intuition will be that spans of the form $[0,2i\rangle$ will correspond to the choices we make, and the spans of the form $[0,2i+1\rangle$ will correspond to choices that we do not make.

  The spanner $r_\valid$ is defined as a VA by first opening all variables $x_{i,j}$. Then, for all $1 \leq i \leq n_x$, considering the $i$-th factor $aa$, we do a disjunction between two possibilities:
  \begin{itemize}
    \item closing the variables $x_{i,j}$ with $j \in T_i$, then reading $a$, then closing the variables $x_{i,j}$ with $j \in [n_c] \setminus T_i$, then reading $a$; or
    \item closing the variables $x_{i,j}$ with $j \in F_i$, then reading $a$, then closing the variables $x_{i,j}$ with $j \in [n_c] \setminus F_i$, then reading $a$.
  \end{itemize}
  Finally, we read~$b$ (note that all variables are closed). Each captured mapping corresponds to an assignment, where we assign variable $x_{i,j}$ to $[0,2i\rangle$ if we have assigned $x_i$ so as to make clause $j$ true, or to $[0,2i+1\rangle$ if we have assigned $x_i$ in a way that does not make clause $j$ true (including if $x_i$ does not occur in clause~$j$).

  The spanner $r_\mask$ is defined as a disjunction of spanners defined as follows for each $1 \leq k \leq [n_c]$. First open all variables of the form $x_{i,k}$. Then, for the $i$-th factor $aa$, read one~$a$ then close $x_{i,k}$  then read one~$a$. Finally on the final $b$ assign the variables $x_{i,j}$ for all $i\in [n_x]$ and all $j \in [n_c]\setminus \{k\}$.

  We claim again that a mapping of $r_\valid$ is strictly dominated by the $k$-th mapping $m_k$ of $r_\mask$ for $1 \leq k \leq [n_c]$ iff it does not assign variables making clause~$k$ true. This is because all variables of the form $x_{i,j}$ with $j\neq k$ were assigned to the final~$b$ of~$d$ in~$m_k$ which strictly dominate all other possible choices; so the only way not to be strictly dominated is to have assigned a variable of the form $x_{i,k}$ elsewhere than where $m_k$ assigns it, i.e., to have satisfied the clause. Hence, a mapping of $r_\valid$ is not strictly dominated by any mapping of $r_\mask$ iff it encodes a satisfying assignment of~$\Phi$.

  We consider again the skyline of $r_\valid \lor r_\mask$, and conclude from what precedes, like in Theorem~\ref{thm:maindicho}, that the skyline extraction problem for that rule cannot be in output-polynomial combined complexity unless $\mathsf{P} = \mathsf{NP}$.
\end{proof}

We note, however, that the \emph{reverse} of that rule, where all strict
domination pairs share the same left-hand-side, is in fact tractable (and also
fails the UMDSDP). This illustrates that, counter-intuitively, a complexity
classification  on variable-wise
domination rules would not be symmetric between the left-hand-side and
right-hand-side:

\begin{propositionrep}
	The skyline extraction problem for the variable-wise domination rule 
  $\dagg{x}\{a^*\} a^* x\{b\} \lor \domrule{self}$ is output-polynomial in combined complexity.
\end{propositionrep}
\begin{proof}
  It is more intuitive to give an algorithm for the mirror rule: $x\{b\} a^* \dagg{x}\{a^*\}$. We do so, because tractability is preserved under taking the mirror image, simply by taking the mirror of the input VA and document. More precisely, to take the mirror of the input VA, we simply reverse all transitions (for each transition going from $q$ to $q'$, we create one instead going from $q'$ to~$q$), we exchange initial states and final states, and last modify the VA to satisfy our definition which only allows one initial state: this can be done by adding $\epsilon$-transitions to the model (see Appendix~\ref{app:prelim}), adding one initial state with $\epsilon$-transitions to all former initial states, and then removing $\epsilon$ transitions. This mirroring process is in polynomial time, and thus we can show tractability for the mirror rule as stated.

  Let us consider an input VA $\calA$ (which we assume without loss of generality to be trimmed)
  with state space~$Q$ 
  and an input document $d$.
  We assume that the input document $d$ is in the language $b a^*$ of the domination rule, as otherwise the domination relationship is trivial and we can simply compute the entire output of $\calA$ on~$d$ in output-polynomial time. Now,
  we make the following claim~(*): each mapping of $\calA(d)$ can dominate at most $|Q|^2$ other mappings. If this is true, then we immediately conclude an output-polynomial algorithm like in the proof of Corollary~\ref{cor:outputpolynomial}.

  To see why claim~(*) is true, let us consider which sets of variables can be assigned to the initial $b$
  by the automaton. Let $q_1, \ldots, q_n$ (with $n \leq |Q|$) be the states that can be reached from the initial state~$q_0$ by reading only markers followed by one $b$-transition. By Observation~\ref{obs:openvariables}, to each state $q_i$ with $1 \leq i \leq n$ must correspond one precise set $X_i$ of open variables that are still open at that stage, because $\calA$ is sequential and trimmed. Now, consider the sets 
$q_1', \ldots, q_m'$ (with $m \leq |Q|$) be the states reachable from the $q_i$ via marker transitions. Again, to each of these states $q_i'$ must correspond one precise set $X_i'$ of variables which still have an open marker.

  Now, an accepting run of $\calA$ must proceed by first going to a state~$q_i$, while reading the initial $b$ and opening a precise set of variables $X_i$, and then continuing via marker transitions to a state~$q_j'$ before making the next letter transition (or accepting, if $d$ is the single-letter document~$b$), where the set of variables that are still open is $X_j'$. The set of variables assigned to the initial $b$ in that run is precisely $X_i \setminus X_j'$, i.e., the variables open just after reading the~$b$ and closed before the next transition. We can now see that, as $n, m \leq |Q|$, there are at most $|Q|^2$ such sets.

  This argument shows that, if a mapping $m$ is strictly dominated by another mapping~$m'$, then by the variable domination rule the difference between $m$ and $m'$ must be that $m$ assigns some variables to the initial~$b$ which $m'$ assigns to a suffix of the form~$a^*$, and the other variables are assigned identically. But there are only $|Q|^2$ possible sets of variables assigned to the first~$b$, so at most $|Q|^2$ such mappings $m$. This establishes our claim (*), and allows us to conclude the proof.
\end{proof}

\section{Conclusions}
\label{sec:conc}
We have introduced the general framework of domination rules to express the skyline operator for document
spanners, with rules that are themselves expressed as a spanner. We have shown
that this operator (with regular rules) does not increase the expressiveness of regular spanners,
but that
it incurs an unavoidable exponential blowup in the state complexity and is
intractable to evaluate in combined complexity for many natural fixed rules.

Our work leaves several questions open for future investigation. The most
immediate question is whether the skyline extraction problem admits a dichotomy
on the variable-wise regular domination rule in the general case, i.e., extending
Theorem~\ref{thm:maindicho} to arbitrary such rules. However, this seems
challenging.
Another question is whether the hardness results of Section~\ref{sec:further}
also give state complexity lower bounds of the kind shown in
Section~\ref{sct:blowups}, in particular in the schema-based context; and
whether there is a dichotomy on state complexity.

Last, an intriguing question is whether the \emph{top-$k$ problem} of computing a constant number~$k$ of mappings from the skyline is always tractable in combined complexity. None of our hardness results precludes it, but we are not aware of an algorithm for that problem.

\bibliographystyle{plainurl}
\bibliography{main}

\end{document}